\documentclass[11pt,a4paper]{article}

\usepackage{amsmath,amsfonts,amssymb,amsthm}
\usepackage{amsthm}
\usepackage{graphicx,color}
\usepackage{boxedminipage}
\usepackage[linesnumbered, boxed, algosection,noline]{algorithm2e}
\usepackage{framed}
\usepackage{thmtools}
\usepackage{thm-restate}
\usepackage{xspace}
\usepackage{vmargin}
\setmarginsrb{1in}{1in}{1in}{1in}{0pt}{0pt}{0pt}{6mm}
\usepackage{todonotes}
 \usepackage[ pdftex, plainpages = false, pdfpagelabels, 
                 bookmarks=false,
                 bookmarksopen = true,
                 bookmarksnumbered = true,
                 breaklinks = true,
                 linktocpage,
                 pagebackref,
                 colorlinks = true,  
                 linkcolor = blue,
                 urlcolor  = blue,
                 citecolor = red,
                 anchorcolor = green,
                 hyperindex = true,
                 hyperfigures
                 ]{hyperref} 
 \usepackage{xifthen}
 \usepackage{tabularx}
\usepackage{tikz} 
 \usetikzlibrary{calc}

\newcommand{\dist}{\textrm{\rm dist}}

\DeclareMathOperator{\operatorClassNP}{NP}
\newcommand{\classNP}{\ensuremath{\operatorClassNP}}

\DeclareMathOperator{\operatorClassFPT}{FPT\xspace}
\newcommand{\classFPT}{\ensuremath{\operatorClassFPT}\xspace}
\DeclareMathOperator{\operatorClassW}{W}
\newcommand{\classW}[1]{\ensuremath{\operatorClassW[#1]}}

\DeclareMathOperator{\operatorClassXP}{XP\xspace}
\newcommand{\classXP}{\ensuremath{\operatorClassXP}\xspace}

\newcommand{\cO}{\mathcal{O}}
\newcommand{\Oh}{\mathcal{O}}

\newcommand{\bran}[1]{branchable\xspace}

\newtheorem{theorem}{Theorem}
\newtheorem{lemma}{Lemma}[section]

\newtheorem{definition}{Definition}[section]
\newtheorem{observation}{Observation}[section]
\newtheorem{proposition}{Proposition}[section]

\theoremstyle{definition}
\newtheorem{reduction}{Reduction Rule}[section]

\newcommand{\pname}{\textsc}
\newcommand{\ProblemFormat}[1]{\pname{#1}}
\newcommand{\ProblemIndex}[1]{\index{problem!\ProblemFormat{#1}}}
\newcommand{\ProblemName}[1]{\ProblemFormat{#1}\ProblemIndex{#1}{}\xspace}

\newcommand{\probWMSC}{\ProblemName{Minimum Spanning Circuit}}
\newcommand{\probSCIR}{\ProblemName{Spanning Circuit}}
\newcommand{\probewmsc}{\ProblemName{Extended Minimum  Circuit}}

\newcommand{\probemc}{\ProblemName{Minimal Terminal Cut}}
\newcommand{\probbemc}{\ProblemName{Border Contractions}}
\newcommand{\probctse}{\ProblemName{Cycle Through Terminals}}

\newcommand{\ctse}{\probctse}
\newcommand{\wmsc}{\probWMSC}
\newcommand{\scir}{\probSCIR}
\newcommand{\emc}{\probemc}
\newcommand{\bemc}{\probbemc}
\newcommand{\ewmsc}{\probewmsc}
\newcommand{\escir}{\textsc{Extended Spanning Circuit}}

\newcommand{\constrext}{circuit constraint\xspace}

\newcommand{\dpair}{{$(T,{\cal M})$}}



\newlength{\RoundedBoxWidth}
\newsavebox{\GrayRoundedBox}
\newenvironment{GrayBox}[1]%
   {\setlength{\RoundedBoxWidth}{.93\textwidth}
    \def\boxheading{#1}
    \begin{lrbox}{\GrayRoundedBox}
       \begin{minipage}{\RoundedBoxWidth}}%
   {   \end{minipage}
    \end{lrbox}
    \begin{center}
    \begin{tikzpicture}%
       \node(Text)[draw=black!20,fill=white,rounded corners,%
             inner sep=2ex,text width=\RoundedBoxWidth]%
             {\usebox{\GrayRoundedBox}};
        \coordinate(x) at (current bounding box.north west);
        \node [draw=white,rectangle,inner sep=3pt,anchor=north west,fill=white] 
        at ($(x)+(6pt,.75em)$) {\boxheading};
    \end{tikzpicture}
    \end{center}}     

\newenvironment{defproblemx}[2][]{\noindent\ignorespaces%
                                \FrameSep=6pt%
                                \parindent=0pt%
                \vspace*{-1.5em}
                \ifthenelse{\isempty{#1}}{%
                  \begin{GrayBox}{\textsc{#2}}%
                }{%
                  \begin{GrayBox}{\textsc{#2} parameterized by~{#1}}%
                }
                \begin{tabular*}{\textwidth}{@{\hspace{.1em}} >{\itshape} p{1.8cm} p{0.8\textwidth} @{}}%
            }{
                \end{tabular*}%
                \end{GrayBox}%
                \ignorespacesafterend
            }  

\newcommand{\defparproblema}[4]{
  \begin{defproblemx}[#3]{#1}
    Input:  & #2 \\
    Task: & #4
  \end{defproblemx}
}%

\newcommand{\defproblema}[3]{
  \begin{defproblemx}{#1}
    Input:  & #2 \\
    Task: & #3
  \end{defproblemx}
}%

\begin{document}

\title{Spanning  Circuits in Regular Matroids
}

\author{
Fedor V. Fomin\thanks{
Department of Informatics, University of Bergen, Norway.} \addtocounter{footnote}{-1}
\and
Petr A. Golovach\footnotemark{}\addtocounter{footnote}{-1}
\and 
Daniel Lokshtanov\footnotemark{}\addtocounter{footnote}{-1}
\and
Saket Saurabh\footnotemark{} \thanks{Institute of Mathematical Sciences, Chennai, India}}

\date{}

\maketitle

\thispagestyle{empty}

\begin{abstract}
We consider the fundamental Matroid Theory problem of finding a circuit in a matroid spanning a set $T$ of given terminal elements. 
For graphic matroids this corresponds to the problem of finding a simple cycle passing through a set of given terminal edges in a graph. 
The  algorithmic study of the problem on  regular matroids, a superclass of graphic matroids, was initiated by Gaven{\v{c}}iak,   Kr{\'a}l',   and Oum [ICALP'12], who proved that the case of the problem with $|T|=2$ is fixed-parameter tractable (FPT) when parameterized by the length of the circuit.  We extend the result of Gaven{\v{c}}iak,  Kr{\'a}l', and Oum by showing that for regular matroids 
\begin{itemize}
 \item the \probWMSC problem, deciding whether there is a circuit with at  most $\ell$ elements   containing $T$, is FPT  
  parameterized by $k=\ell-|T|$;
\item  the \probSCIR problem, deciding whether there is a circuit containing $T$, is FPT parameterized by $|T|$.
 \end{itemize}
We note that extending our algorithmic findings to binary matroids, a superclass of regular matroids, is highly unlikely:  \probWMSC parameterized by $\ell$ is \classW{1}-hard on binary matroids even when $|T | = 1$. 
We also show a limit to how far our results can be strengthened by considering a smaller parameter.
 More precisely, we prove that \probWMSC parameterized by $|T|$ is \classW{1}-hard even on cographic matroids, a proper subclass of regular matroids.
\end{abstract}

\section{Introduction}\label{sec:Intro}
Deciding if a given graph $G$ contains a cycle passing through a specified set $T$ of terminal edges or vertices  is the classical problem in graph theory. The study of this problem can be traced back to the fundamental theorem of Dirac from 1960s about the existence of a cycle in $k$-connected graph passing through a given set of $k$ vertices \cite{MR0121311}.  According to Kawarabayashi \cite{Kawarabayashi08}  \emph{``...cycles through a vertex set or an edge set are one of central topics in all of graph theory."} 
We refer to \cite{Kawarabayashi02a} for an overview on the graph-theoretical study of the problem, including the famous Lov{\'{a}}sz-Woodall Conjecture. 

The algorithmic version of this question, is there a polynomial time algorithm deciding if a given graph contains a cycle passing through the set of terminal vertices or edges,  is the problem  of a fundamental importance  in graph algorithms. Since the problem generalizes the classical Hamiltonian cycle problem, it is NP-complete. However, for a fixed number of  terminals the problem is solvable in polynomial time. 
The case $|T|=1$ with one terminal vertex or edge is trivially  solved by the breadth first search. The case of $|T|=2$ can be reduced to  finding a flow of size 2 between  two vertices in a graph. The case of $|T|=3$ is already nontrivial and was shown to be solvable in linear time
in \cite{LaPaughR80}, see also \cite{FleischnerW92}. 
The fundamental result of  Robertson and Seymour on the disjoint path problem \cite{RobertsonS-GMXIII}
implies that the problem can be solved in polynomial time for a fixed number of terminals. 
 Kawarabayashi  in \cite{Kawarabayashi08} provided a quantitative improvement by showing that the problem is solvable in polynomial time for $|T|=\cO((\log \log{n})^{1/10})$, where $n$ is the size of the input graph.   
 Bj{\"{o}}rklund et al. \cite{BjorklundHT12} gave a randomized algorithm solving the problem in time 
 $2^{|T|}n^{\cO(1)}$.   The algorithm of Bj{\"{o}}rklund et al. solves  also the minimization variant of the problem, where 
 the task is to find a cycle of minimum length passing through terminal vertices.  We refer to the book of Cygan et al.~\cite{CyganFKLMPPS15} for an overview of different techniques in parameterized algorithms for solving problems about cycles and paths in graphs. 
 
 Matroids are combinatorial objects generalizing graphs and linear independence. The study of circuits  containing certain elements of a matroid  is  
one of the central themes  in  matroid theory. 
 For graphic matroids, the problem of finding a circuit spanning (or containing) a given set of elements  corresponds to   finding in a   graph a simple cycle passing through specified edges.
The classical theorem of Whitney  \cite{Whitney35} asserts that any pair of elements of a connected matroid are in a circuit. Seymour   \cite{SeymourP86} obtained a characterization of binary matroids with a circuit containing a  triple of elements. See also 
  \cite{denley2001generalization,mcguinness2009ore,Oxley97} and references there for combinatorial results about circuits spanning certain elements in matroids. 
However,  compared to graphs,  the algorithmic aspects  of ``circuits through elements"  in matroids are much less understood.

 In their work on deciding first order properties on matroids of locally bounded  branch-width, Gaven{\v{c}}iak et al. 
\cite{GavenciakKO12} initiated the algorithmic study of the following problem.

\defproblema{\probWMSC}%
{A binary matroid $M$ with a ground set $E$, a weight function $w\colon E\rightarrow \mathbb{N}$,
a set of \emph{terminals} $T\subseteq E$, and a nonnegative integer $\ell$.}%
{Decide whether there is  a circuit $C$ of $M$ with $w(C)\leq \ell$ such that $T\subseteq C$.}

Since graphic matroids are binary, this problem is a generalization of the problem of finding a cycle through a given set of edges in a graph.  By the result of Vardy  \cite{vardy1997algorithmic} about the  {\textsc{Minimum Distance}} problem from coding theory,   \probWMSC is \classNP-complete even when $T=\emptyset$.
Gaven{\v{c}}iak et al.  \cite{GavenciakKO12} observed that  the hardness result of  
Downey et al. from~\cite{DowneyFVW99}  also implies that \wmsc{} is \classW{1}-hard on binary matroids with unit-weights elements when parameterized by $\ell$ even if $|T|=1$.  Parameterized complexity of \wmsc  for  $T=\emptyset$ on binary matroids, i.e. the case when we ask about the existence of a circuit of length at most $\ell$,  is known as \textsc{Even Set} in parameterized complexity and is a long standing open problem in the area.
The intractability of the problem changes when we restrict the input binary matroid to be regular, i.e. matroid which has a representation by rows of a 
 totally unimodular matrix.  
In particular,  Gaven{\v{c}}iak et al. show that  for $|T|=2$, \wmsc is fixed parameter tractable (\classFPT) being parameterized by $\ell$ by giving time 
$\ell^{\ell^{\ell^{\Oh(\ell)}}}n^{\cO(1)}$ algorithm,  where $n$ is the number of elements in the input matroid. 
Recall that all graphic and cographic matroids are regular and thus  algorithmic results for regular matroids yield algorithms on graphic and cographic matroids. 

\medskip
\noindent\textbf{Our results.} 
In this work we  show, and this is the main result of the paper,   that on regular matroids  \wmsc is  \classFPT being parameterized by $\ell$ without any additional condition on the size of the terminal set. Actually, we obtain the algorithm for  ``stronger" parameterization  $k=\ell-w(T)$. The running time of our algorithm is  $2^{\Oh(k^2\log k)}\cdot n^{\Oh(1)}$.

Our approach  is based on the classical decomposition theorem of Seymour \cite{Seymour81}. Roughly speaking, the theorem allows to decompose a regular matroid by making use of 1,2, and 3-sums into graphic, cographic matroids and matroid of a fixed size. (We refer to Section~\ref{sec:regular}
for the  precise formulation of the theorem).  
Thus to solve the problem on regular matroids, one has to understand how to solve a certain extension of the problem  on graphic and cographic matroids (matroids of constant size are usually trivial), and then employ Seymour's theorem to combine solutions. 
This is exactly the  approach which was taken by  Gaven{\v{c}}iak et al.  in   \cite{GavenciakKO12} for solving the problem for  $|T|=2$, and this is the approach we adapt in this paper.   However,  the details are very different.  
In particular, in order to use the general framework,  we have  to  solve the problem on cographic matroids, which is already quite non-obvious. 
Gaven{\v{c}}iak et al.  \cite{GavenciakKO12} adapt  the method of Kawarabayashi and Thorup  \cite{KawarabayashiT11} who used it to prove that
finding an edge-cut with at most $s$ edges that separates the input graph into at least $k$ component is \classFPT{} when parameterized by $s$.
This approach works for $|T|=2$ and probably may be extended for the case when the number of terminals is bounded, but we doubt that it could be applied for the parameterization by  $k=\ell-w(T)$. Hence,  in order to solve  \wmsc on cographic matroids, 
 we use the recent framework of \emph{recursive understanding} developed by Chitnis et al. in~\cite{ChitnisCHPP12} for the \emc problem. In this problem,  we are given a a connected graph
$G$  with a terminal set of  edges  $T\subseteq E(G)$ and  terminal vertex sets $R_1,R_2\subseteq V(G)$, and  
the task is  to find a cut $C$ of small weight satisfying a number of constraints: (a) this cut should be a minimal cut-set, (b) it should contain all edges of $T$, and (c)  it should separate $R_1$ from $R_2$, meaning that  $G-C$ contains   distinct connected components $X_1$ and $X_2$ such that $R_i\subseteq X_i$ for $i\in\{1,2\}$.
We believe that this problem is interesting on its own. 
Finally, constructing a solution by going through Seymour's matroid decomposition 
when $|T|$ is unbounded 
is also a non-trivial procedure requiring a careful analyses.

In our case, in order to solve  \wmsc on cographic matroids, we have to settle the following  problem. In the \emc problem,  we are given a a connected graph
$G$  with a terminal set of  edges  $T\subseteq E(G)$ and  terminal vertex sets $R_1,R_2\subseteq V(G)$. 
The task is  to find a cut $C$ of small weight satisfying a number of constraints: (a) this cut should be a minimal cut-set, (b) it should contain all edges of $T$, and (c)  it should separate $R_1$ from $R_2$, meaning that  $G-C$ contains   distinct connected components $X_1$ and $X_2$ such that $R_i\subseteq X_i$ for $i\in\{1,2\}$.
This problem is interesting on its own. The solution to this problem is 
non-trivial and here we use the recent framework of \emph{recursive understanding} developed by Chitnis et al. in~\cite{ChitnisCHPP12}. Finally, constructing a solution by going through Seymour's matroid decomposition, is also a non-trivial  procedure requiring a careful analyses. 

With a similar approach, we also obtain an algorithm for the following decision version of the problem, where we put no constrains on the size of the circuit.

 \defproblema{\probSCIR}{A binary matroid $M$ with a ground set $E$ and a set of terminals $T\subseteq E$.} 
{Decide whether there is a circuit $C$ of $M$  such that $T\subseteq C$.}
 
 We show that on regular matroids \probSCIR is \classFPT parameterized by $|T|$.

 \medskip
 
 The remaining part of the paper is organized as follows. In Section~\ref{sec:defs} we introduce basic notions used in the paper. In Section~\ref{sec:regular} we briefly introduce the fundamental structural results of Seymour~\cite{Seymour80a} about regular matroids. We also explain the refinement of the decomposition theorem of Seymour~~\cite{Seymour80a}  given by Dinitz and Kortsarz in~\cite{DinitzK14} that is more convenient for the algorithmic purposes. We conclude this section by some structural results about circuits in regular matroids. Section~\ref{sec:cut-graph} contains the algorithm for \emc{}. In Section~\ref{sec:wmsc} we give the  algorithm for~\wmsc{} on regular matroids. First, we solve the extended variant of \wmsc{} on matroids that are basic for the Seymour's decomposition~\cite{Seymour80a}. Then, we explain how to obtain the general result. We follow the same scheme in Section~\ref{sec:scir} for \scir{} parameterized by $|T|$. In Section~\ref{sec:conclusion} we provide some hardness observations and state open problems.
 
\section{Preliminaries}\label{sec:defs}
{\bf Parameterized Complexity.}
Parameterized complexity is a two dimensional framework
for studying the computational complexity of a problem. One dimension is the input size
$n$ and another one is a parameter $k$. It is said that a problem is \emph{fixed parameter tractable} (or \classFPT), if it can be solved in time $f(k)\cdot n^{O(1)}$ for some function $f$.  We refer to the recent books of Cygan et al.~\cite{CyganFKLMPPS15} and  Downey and Fellows~\cite{DowneyF13} for  an introduction  to parameterized complexity. 

It is standard for a parameterized algorithm to use \emph{(data) reduction rules}, i.e., polynomial  or \classFPT{} algorithms that either solve an instance or reduce it to another one that typically has a lesser input size and/or a lesser value of the parameter. 
We say that  reduction rule is \emph{safe} if it either correctly solves the problem or outputs an equivalent instance of the problem without increasing the parameter.

\medskip
\noindent
{\bf Graphs.} 
We consider finite undirected (multi) graphs that can have  loops or multiple edges. Throughout the paper we use $n$ to denote the number of vertices and $m$ the number of edges of considered graphs unless it crates confusion.
For a graph $G$ and a subset $U\subseteq V(G)$ of vertices, we write $G[U]$ to denote the subgraph of $G$ induced by $U$. We write $G-U$ to denote the subgraph of $G$ induced by $V(G)\setminus U$, and $G-u$ if $U=\{u\}$.
Respectively, for $S\subseteq E(G)$, $G[S]$ denotes the graph induced by $S$, i.e., the graph with the set of edges $S$ whose vertices are the vertices of $G$ incident to the edges of $S$.
We denote by $G-S$ the graph obtained from $G$ by the deletion of the edges of $G$; for a single element set, we write $G-e$ instead of $G-\{e\}$.
For $e\in E(G)$, we denote by $G/e$ the graph obtained by the contraction of $e$. Since we consider multigraphs, it is assumed that if $e=uv$, then to construct $G/e$, we delete $u$ and $v$, construct a new vertex $w$, and then for each $ux\in E(G)$ and each $vx\in E(G)$, where $x\in V(G)\setminus \{u,v\}$, we construct new edge $wx$ (and possibly obtain multiple edges), and for each $e'=uv\neq e$, we add a new loop $ww$.
For a vertex $v$, we denote by $N_G(v)$ the \emph{(open) neighborhood} of $v$, i.e., the set of vertices that are adjacent to $v$ in $G$.
For a set $S\subseteq V(G)$, $N_G(S)=(\cup_{v\in S}N_G(v))\setminus S$. We denote by $N_G[v]=N_G(v)\cup\{v\}$ the \emph{closed neighborhood} of $v$. To vertices $u$ and $v$ are \emph{true twins} if $N_G[u]=N_G[v]$, and $u$ and $v$ are \emph{false twins} if $N_G(u)=N_G(v)$.

\medskip
\noindent
{\bf Cuts.} Let $G$ be a graph. 
A \emph{cut} $(A,B)$ of a graph $G$ is a partition of $V(G)$ into two disjoint sets $A$ and $B$. 
A set $S\subseteq E(G)$ is an \emph{(edge) cut-set} if the deletion of $S$ increases the number of components. A cut-set $S$ is \emph{(inclusion) minimal} if any proper subset of $S$ is not a cut-set. 
A \emph{bridge} is a cut-set of size one. 
For two disjoint vertex sets of vertices $A$ and $B$ of a graph $G$, $E(A,B)=\{uv\in E(G)\mid u\in A,v\in B\}$.
Clearly, $E(A,B)$ is an edge  cut-set, and for any cut-set $S\subseteq E(G)$, there is a cut $(A,B)$ with $S=E(A,B)$. Notice also that
$E(A,B)$ is a minimal cut-set of  a connected graph $G$ if and only if $G[A]$ and $G[B]$ are connected.

\medskip
\noindent
{\bf Matroids.}
We refer to the book of Oxley~\cite{Oxley92} for the detailed introduction to matroid theory.
 Recall that a matroid $M$ is a pair $(E,\mathcal{I})$, where $E$ is a finite \emph{ground} set of $M$ and $\mathcal{I}\subseteq 2^E$ is a collection of \emph{independent} sets that satisfy the following three axioms:
\begin{itemize}
\item[I1.] $\emptyset\in \mathcal{I}$,
\item[I2.] if $X\in \mathcal{I}$ and $Y\subseteq X$, then $Y\in\mathcal{I}$,
\item[I3.] if $X,Y\in \mathcal{I}$ and $|X|<|Y|$, then there is $e\in Y\setminus X$ such that $X\cup\{e\}\in \mathcal{I}$.
\end{itemize}
We denote the ground set of $M$ by $E(M)$
and 
the set of independent set by $\mathcal{I}(M)$ or simply by $E$ and  $\mathcal{I}$ if it does not creates confusion. If a set $X\subseteq E$ is not independent, then $X$ is \emph{dependent}. 
Inclusion maximal independent sets are called \emph{bases} of $M$. We denote the set of bases by $\mathcal{B}(M)$ (or simply by $\mathcal{B}$).
The matroid $M^*$ with the ground set $E(M)$ such that $\mathcal{B}(M^*)=\mathcal{B}^*(M)=\{E\setminus B\mid B \in\mathcal{B}(M) \}$ is \emph{dual} to $M$. 
 
 An (inclusion) minimal dependent set is called a \emph{circuit} of $M$. We denote the set of all circuits of $M$ by $\mathcal{C}(M)$ or simply $\mathcal{C}$ if it does not create a confusion.  
The circuits satisfy the following conditions (\emph{circuit axioms}):
\begin{itemize}
\item[C1.] $\emptyset\notin \mathcal{C}$,
\item[C2.] if $C_1,C_2\in \mathcal{C}$ and $C_1\subseteq C_2$, then $C_1=C_2$,
\item[C3.] if $C_1,C_2\in \mathcal{C}$, $C_1\neq C_2$, and $e\in C_1\cap C_2$, then there is $C_3\in \mathcal{C}$ such that $C_3\subseteq (C_1\cup C_2)\setminus\{e\}$.
\end{itemize}  
An one-element circuit is called \emph{loop}, and if $\{e_1,e_2\}$ is a two-element circuit, then it is said that $e_1$ and $e_2$ are \emph{parallel}. An element $e$ is \emph{coloop} if 
$e$ is a loop of $M^*$ or, equivalently,
$e\in B$ for every $B\in\mathcal{B}$.  A \emph{circuit} of $M^*$ is called \emph{cocircuit} of $M$. 
 A set $X\subseteq E$ is a \emph{cycle} of $M$ if $X$ either empty or $X$ is a disjoint union of circuits. By $\mathcal{S}(M)$ (or $\mathcal{S}$) we denote the set of all cycles of $M$. 
The sets of circuits and cycles completely define matroid. Indeed, a set is independent if and only if it does not contain a circuit, and the circuits are exactly inclusion minimal nonempty cycles.

Let $M$ be a matroid, $e\in E(M)$. The matroid $M'=M-e$ is obtained by \emph{deleting} $e$ if $E(M')=E(M)\setminus\{e\}$ and $I(M')=\{X\in \mathcal{I}(M)\mid e\notin X\}$. 
We say that $M'$ is obtained from $M$ by \emph{adding a parallel to $e$  element} if $E(M')=E(M)\cup\{e'\}$, where $e'$ is a new element, and 
$\mathcal{I}(M')=\mathcal{I}(M)\cup \{(X\setminus\{e\})\cup\{e'\}\mid X\in \mathcal{I}(M)\text{ and }e\in X\}$. It is straightforward to verify that $\mathcal{I}(M')$ satisfies the axioms I.1-3, i.e., $M'$ is a matroid with the ground set $E(M)\cup\{e'\}$. It is also easy to see that $\{e,e'\}$ is a circuit, that is, $e$ and $e'$ are parallel elements of $M'$. 

We can observe the following.

\begin{observation}\label{obs:par}
Let $\{e_1,e_2\},C\in \mathcal{C}$ for a matroid $M$. If $e_1\in C$ and $e_2\notin C$, then $C'=(C\setminus\{e_1\})\cup\{e_2\}$ is a circuit.
\end{observation}
 
\begin{proof} 
By the axiom C3, $(\{e_1,e_2\}\cup C)\setminus\{e_1\}=(C\setminus\{e_1\})\cup\{e_2\}=C'$ contains a circuit $C''$. Suppose that $C''\neq C'$. Notice that because $C\setminus \{e_1\}$ contains no circuit, 
$e_2\in C''$. As $e_1\notin C''$, we obtain that $(\{e_1,e_2\}\cup C'')\setminus\{e_2\}$ contains a circuit, but $(\{e_1,e_2\}\cup C'')\setminus\{e_2\}$ is a proper subset of $C$; a contradiction. Hence, $C''=C'$, i.e., $C'$ is a circuit. 
\end{proof}

\medskip
\noindent
{\bf Matroids associated with graphs.} Let $G$ be a graph. The \emph{cycle} matroid $M(G)$ has the ground set $E(G)$ and a set $X\subseteq E(G)$ is independent if $X=\emptyset$ or $G[X]$ has no cycles.  Notice that $C$ is a circuit of $M(G)$ if and only if $C$ induces a cycle of $G$.  The \emph{bond} matroid $M^*(G)$  with the ground set $E(G)$ is dual to $M(G)$, and  $X$ is a circuit of $M^*(G)$ if and only if $X$ is a minimal cut-set of $G$. Respectively, \wmsc{} for a cycle matroid $M(G)$ is to decide  whether $G$ has a cycle $C$ of weight at most $\ell$ that goes through the edges of $T$, and  for a bond matroid $M^*(G)$ it is to decide whether $G$ has a minimal cut-set $C$ of weight at most $\ell$ that contains $T$. 
We say that $M$ is a \emph{graphic} matroid if $M$ is isomorphic to $M(G)$ for some graph $G$. Respectively, $M$ is \emph{cographic} if 
there is graph $G$ such that  $M$ is isomorphic to $M^*(G)$.
Notice that $e\in E$ is a loop of a cycle matroid $M(G)$ if and only if $e$ is a loop of $G$, and 
$e$ is a loop of $M^*(G)$ if and only if $e$ is a bridge of $G$.

Notice also that by the addition of an element parallel to $e\in E$ for $M(G)$ we obtain $M(G')$ for the graph $G'$ obtained by adding a new edge with the same end vertices as $e$. Respectively,  by adding of an element parallel to $e\in E$ for $M^*(G)$ we obtain $M^*(G')$ for the graph $G'$ obtained by subdividing $e$. Hence, adding or deleting a parallel element of graphic or cographic matroid does not put it outside the corresponding class.

\medskip
\noindent
{\bf Matroid representations.}  Let $M$ be a matroid and let $F$ be a field.  An $n\times m$-matrix $A$ over $F$ is a \emph{representation of $M$ over $F$} if there is one-to-one correspondence $f$ between $E$ and the set of columns of $A$ such that for any $X\subseteq E$, $X\in \mathcal{I}$ if and only if the columns $f(X)$ are linearly independent (as vectors of $F^n$); if $M$ has such a representation, then it is said that $M$ has a \emph{representation over $F$}. In other words, $A$ is a representation of $M$ if $M$ is isomorphic to the \emph{column matroid} of $A$, i.e., the matroid whose ground set is the set of columns of $A$ and a set of columns is independent if and only if these columns are linearly independent. 
A matroid is \emph{binary} if it can be represented over ${\rm GF}(2)$. A matroid is \emph{regular} if it can be represented over any field. In particular, graphic and cographic matroids are regular.

 As we are working with binary matroids, we assume that for an input matroid, we are given  its representation over ${\rm GF}(2)$. 
 Then it can be checked in polynomial time whether a subset of the ground set is independent by checking the linear independence of the corresponding columns.

\section{Structure of regular matroids}\label{sec:regular}
 Our results for regular matroids use  the structural decomposition for regular matroids given by  
Seymour~\cite{Seymour80a}. Recall that, for two set $X$ and $Y$, $X\bigtriangleup Y=(X\setminus Y)\cup (Y\setminus X)$ denotes the \emph{symmetric difference} of $X$ and $Y$. For our purpose we also need the following observation.

\begin{observation}[see~\cite{Oxley92}]\label{obs:symm}
Let $C_1$ and $C_2$ be circuits (cycles) of a binary matroid $M$. Then $C_1\bigtriangleup C_2$ is a cycle of $M$.
\end{observation}

To describe the decomposition of matroids we need the notion of ``$r$-sums''  of matroids. However for our purpose it is sufficient that we restrict ourselves to binary matroids and  up to $3$-sums. We refer to~\cite[Chapter 8]{Truemper92} for a more detailed introduction to matroid sums. 
Let $M_1$ and $M_2$ be binary matroids. The \emph{sum} of $M_1$ and $M_2$, denoted by $M_1\bigtriangleup M_2$, is the matroid $M$ with the ground set $E(M_1)\bigtriangleup E(M_2)$. The cycles of $M$ are all subsets $C\subseteq E(M_1)\bigtriangleup E(M_2)$ of the form  $C_1\bigtriangleup C_2$, where $C_1$ is a cycle of $M_1$ and  $C_2$ is a cycle of $M_2$.
This does indeed define a binary matroid~\cite{Seymour80a} 
as can be seen from Observation~\ref{obs:symm}, 
in which the circuits are the minimal nonempty cycles and the independent sets are (as always) the sets that do not contain any circuit. 
For our purpose the following special cases of matroid sums are sufficient. 
\begin{enumerate}
\setlength{\itemsep}{-2pt}
\item  If $E(M_1)\cap E(M_2)=\emptyset$ and  $E(M_1), E(M_2)\neq\emptyset$, then  $M$ is the \emph{$1$-sum} of $M_1$ and $M_2$ and we write $M=M_1\oplus_1 M_2 $.
\item  If $|E(M_1)\cap E(M_2)|=1$, the unique $e\in E(M_1)\cap E(M_2)$ is not a loop or coloop of $M_1$ or $M_2$, and $|E(M_1)|,|E(M_2)|\geq 3$, then $M$ is the \emph{$2$-sum} of $M_1$ and $M_2$ and we write $M=M_1\oplus_2 M_2$.
\item  If $|E(M_1)\cap E(M_2)|=3$, the 3-element set $Z=E(M_1)\cap E(M_2)$ is a circuit of $M_1$ and $M_2$, $Z$ does not contain a cocircuit of $M_1$ or $M_2$, 
and $|E(M_1)|,|E(M_2)|\geq 7$, 
then $M$ is the \emph{$3$-sum} of $M_1$ and $M_2$ and we write $M=M_1\oplus_3 M_2$.
\end{enumerate}
If $M=M_1\oplus_rM_2$ for some $r\in\{1,2,3\}$, then we write $M=M_1\oplus M_2$.

\begin{definition}
A \emph{$\{1,2,3\}$-decomposition} of a matroid $M$ is a collection of matroids $\mathcal{M}$, called the \emph{basic matroids} and a rooted binary tree $T$ in which $M$ is the root and the elements of $\mathcal{M}$ are the leaves such that any internal node is either $1$-, $2$- or $3$-sum of its children. 
\end{definition}

We also need the special binary matroid $R_{10}$ to be able to define the decomposition theorem for regular matroids. It is represented over   ${\rm GF}(2)$ by the 
$5\times 10$-matrix whose columns are formed by vectors that have exactly three non-zero entries (or rather three ones) and no two columns are identical.  
Now we are ready to give the 
 decomposition theorem for regular matroids due to Seymour~\cite{Seymour80a}.

\begin{theorem}[\cite{Seymour80a}]\label{thm:decomp}
Every regular matroid $M$ has an  $\{1, 2, 3\}$-decomposition in which every basic matroid is either graphic, cographic, or isomorphic to $R_{10}$. Moreover,
such a decomposition (together with the graphs whose cycle and bond matroids are isomorphic to the corresponding basic graphic and cographic matroids) can be found in time polynomial in $|E(M)|$.
\end{theorem}

For our algorithmic purposes we will not use the Theorem~\ref{thm:decomp} but rather a modification 
proved by Dinitz and Kortsarz in~\cite{DinitzK14}. Dinitz and Kortsarz in~\cite{DinitzK14} observed that  some restrictions in the definitions of $2$- and $3$-sums are not important for the algorithmic purposes. In particular, in the definition of the $2$-sum, the unique $e\in E(M_1)\cap E(M_2)$ is not a loop or coloop of $M_1$ or $M_2$, and $|E(M_1)|,|E(M_2)|\geq 3$ could be dropped. Similarly, in the definition of $3$-sum the conditions  that $Z=E(M_1)\cap E(M_2)$ does not contain a cocircuit of $M_1$ or $M_2$, and $|E(M_1)|,|E(M_2)|\geq 7$ could be dropped. We define \emph{extended} $1$-, $2 $- and $3$-sums by omitting these restrictions.  Clearly, Theorem~\ref{thm:decomp} holds if we replace sums by extended sums in the definition of the $\{1,2,3\}$-decomposition. To simplify notations, we use $\oplus_1,\oplus_2,\oplus_3$ and $\oplus$ to denote these extended sums. Finally, we also need the notion of a conflict graph associated with a \emph{$\{1,2,3\}$-decomposition} of a 
matroid $M$ given by  Dinitz and Kortsarz in~\cite{DinitzK14}.

\begin{definition}[\cite{DinitzK14}]
Let \dpair\ be a $\{1, 2, 3\}$-decomposition of  a matroid $M$. The \emph{intersection} (or \emph{conflict}) graph of \dpair\  is the graph $G_T$ with the vertex set $\mathcal{M}$ such that distinct $M_1,M_2\in \mathcal{M}$ are adjacent in $G_T$ if and only if $E(M_1)\cap E(M_2)\neq\emptyset$. 
\end{definition}

Dinitz and Kortsarz in~\cite{DinitzK14} showed how to modify a given decomposition in order to make the conflict graph a forest. In fact they proved a slightly stronger condition that for any $3$-sum (which by definition is summed along a circuit of size $3$), the circuit in the intersection is contained entirely in two of the lowest-level matroids. In other words, while the process of summing matroids might create new circuits that contain elements that started out in different matroids, any circuit that is used as the intersection of a sum existed from the very beginning.

We state the result of \cite{DinitzK14} in the following form that is convenient for us.

\begin{theorem}[\cite{DinitzK14}]\label{thm:decomp-good}
For a given regular matroid $M$, there is a (conflict) tree $\mathcal{T}$, whose set of nodes is a set of matroids $\mathcal{M}$, where each element of $\mathcal{M}$ is a graphic or cographic matroid, or a matroid obtained from $R_{10}$ by (possible) deleting some elements and  adding  parallel elements,  that has the following properties:
\begin{itemize}
\item[i)]  if two distinct matroids $M_1,M_2\in \mathcal{M}$ have nonempty intersection, then $M_1$ and $M_2$ are adjacent in $\mathcal{T}$,
\item[ii)] for any distinct $M_1,M_2\in \mathcal{M}$, $|E(M_1)\cap E(M_2)|=0,~1$ or $3$,
\item[iii)] $M$ is obtained by the consecutive performing extended 1, 2 or 3-sums for adjacent matroids in any order. 
\end{itemize}
Moreover, $\mathcal{T}$ can be constructed in a polynomial time.
\end{theorem}

If $\mathcal{T}$ is a conflict tree for a matroid $M$, we say that $M$ is defined by $\mathcal{T}$. 

In our algorithms we are working with rooted conflict trees. Fixing a  {\em root} $r$ in  $\mathcal{T}$  defines the natural parent-child, descendant and ancestor relationships on the nodes of $\mathcal{T}$.  Our algorithms are based on performing {\em bottom-up} traversal of the tree $\mathcal{T}$. 
We say that a node $M_\ell$ of $\mathcal{T}$ is a \emph{leaf}  if it has no children, and $M_s$ is a \emph{sub-leaf} if it has at least one child and the children of $M_s$ are leaves. Let $M_\ell$ be a leaf and let $M_s$ be its adjacent sub-leaf.  We say that $M_\ell$ is \emph{$s$-leaf} for $s\in\{1,2,3\}$ if the edge between $M_s$ and $M_\ell$ corresponds to the extended $s$-sum. 

\medskip
As in \wmsc{} and \scir{} we are looking for circuits containing terminals, we need some results about the structure of circuits of  matroids and matroid sums.

\begin{lemma}\label{lem:circ-triangle}
Let $Z=\{e_1,e_2,e_3\}$ be a circuit of a binary matroid $M$. Let also $C$ be a circuit of $M$ such that $C\cap Z=\{e_3\}$. If $C'=C\bigtriangleup Z$ is not a circuit, then $C'$ is a disjoint union of two circuits $C_1$ and $C_2$ containing $e_1$ and $e_2$ respectively, and $C_1\bigtriangleup Z$ and  $C_2\bigtriangleup Z$ are circuits.
\end{lemma}

\begin{proof}
 By Observation~\ref{obs:symm},  $C'$ is a cycle of $M$. If $C'$ is not a circuit, then $C'$ is a disjoint union of circuits of $M$. If $C'$ contains a circuit $C''$ such that $C''\cap Z=\emptyset$, then $C''\subset C$ contradicting the condition that $C$ is a minimal dependent set. Hence, each circuit of $C'$ contains an element of $Z$. Since $Z\cap C'=\{e_1,e_2\}$,  $C'$ is a disjoint union of two circuits $C_1$ and $C_2$ containing $e_1$ and $e_2$ respectively. 

Suppose that, say $C_1\bigtriangleup Z$, is not a circuit. 
Then by the above, $C_1\bigtriangleup Z$ is a disjoint union of two circuits $C_2'$ and $C_3'$ containing $e_2$ and $e_3$ respectively. But then $C''=C_2\bigtriangleup C_2'$ is a cycle and $C''\subset C$ contradicting that $C$ is a circuit. Hence, $C_1\bigtriangleup Z$ and  $C_2\bigtriangleup Z$ are circuits.
\end{proof}

\begin{lemma}\label{lem:circ-triangle-a}
Let $Z=\{e_1,e_2,e_3\}$ be a circuit of a binary matroid $M$. Let also $C$ be a circuit of $M$ such that $C\cap Z=\{e_1,e_2\}$. Then $C'=C\bigtriangleup Z$ is a circuit of $M$.
\end{lemma}

\begin{proof}
By Observation~\ref{obs:symm}, $C'$ is a cycle of $M$.  Because $e_3\in C'$, there is a circuit $C''\subseteq C'$ containing $e_3$. If $C''\neq C'$, then the cycle $C''\bigtriangleup Z\subset C$ contradicting the fact that $C$ is a circuit. Hence, $C'=C''$, i.e., $C'$ is a circuit. 
\end{proof}

\begin{lemma}\label{lem:circ}
Let $M=M_1\oplus_r M_2$ for $r\in\{1,2,3\}$, where $M_1$ and $M_2$ are binary matroids, and $Z=E(M_1)\cap E(M_2)$. 
\begin{itemize}
\item[i)] If $r=1$, then $\mathcal{C}(M)=\mathcal{C}(M_1)\cup \mathcal{C}(M_2)$.
\item[ii)] If $r=2$ and $Z=\{e\}$, then 
\begin{align*}
\mathcal{C}(M)&=\{C\in \mathcal{C}(M_1)\mid e\notin C\}\cup \{C\in\mathcal{C}(M_2)\mid e\notin C\}\\
&\cup \{C_1\bigtriangleup C_2\mid C_1\in \mathcal{C}(M_1),C_2\in \mathcal{C}(M_2),e\in C_1,e\in C_2 \}.
\end{align*}
\item[iii)] If $r=3$, then
\begin{align*}
\mathcal{C}(M)=&\{C\in \mathcal{C}(M_1)\mid Z\cap C=\emptyset\}\cup \{C\in\mathcal{C}(M_2)\mid Z\cap C=\emptyset\}\\
\cup& \{C_1\bigtriangleup C_2\mid C_1\in \mathcal{C}(M_1),C_2\in \mathcal{C}(M_2),C_1\cap Z=\{e\}\text{ and }C_2\cap Z=\{e\}\\
&\text{ for some }e\in Z,\text{ and }C_1\bigtriangleup Z\in\mathcal{C}(M_1)\text{ or }C_2\bigtriangleup Z\in\mathcal{C}(M_2)\}.
\end{align*}
\end{itemize}
 
\end{lemma}

\begin{proof}
The claims i) and ii) follow directly from the definitions of the extended 1 and 2-sums. Hence, we have to prove only iii). Recall that $Z$ is a circuit of $M_1$ and $M_2$ in the case of the extended 3-sum.

Let $C$ be  a circuit of $M$. If $C\subseteq E(M_i)$ for $i\in\{1,2\}$, then $C$ is a cycle of $M_i$ and, by minimality, $C$ is a circuit of $M_i$. Assume that $C\setminus E(M_i)\neq\emptyset$ for each $i\in\{1,2\}$. By  definition, $C=C_1\bigtriangleup C_2$ and $C_1\cap Z=C_2\cap Z$, where $C_1$ and $C_2$ are cycles of $M_1$ and $M_2$ respectively.   

If $Z\subseteq E(C_1)$, then  by Observation~\ref{obs:symm}, $C'=C_1\bigtriangleup Z\subseteq C$ is a cycle of $M_1$. Hence, $C'$ is a cycle of $M$ contradicting that $C$ is a  minimal  dependent set. Therefore $1\leq |C_1\cap Z|\leq 2$. Suppose that $|C_1\cap Z|=2$. Consider $C_1'=C_1\bigtriangleup Z$ and $C_2'=C_2\bigtriangleup Z$. By  Observation~\ref{obs:symm}, $C_i'$ ia a  cycle of $M_i$,   $i\in\{1,2\}$. Clearly, $C=C_1'\bigtriangleup C_2'$, but now $|C_1'\cap Z|=|C_2'\cap Z|=1$.
It means, that we always can assume that  $C=C_1\bigtriangleup C_2$, where $C_1\cap Z=\{e\}$ and $C_2\cap Z=\{e\}$ for some $e\in Z$.

Suppose that one of the cycles $C_1$ and $C_2$, say $C_1$, is not a circuit. Then $C_1$ is a disjoint union of circuits of $M_1$. This union contains a circuit $C_1'$ with $e\in C_1'$. Then $C'=C_1'\bigtriangleup C_2\subset C$ is a cycle of $M$ contradicting the minimality of $C$. Hence, $C_1$ and $C_2$ are circuits of $M_1$ and $M_2$ respectively.

Suppose that  $C_1'=C_1\bigtriangleup Z$ and $C_2'=C_2\bigtriangleup Z$ are not circuits  of $M_1$ and $M_2$ respectively. By Lemma~\ref{lem:circ-triangle}, for $i\in\{1,2\}$, $C_i'$ is a disjoint union of two circuits $C_i^{1}$ and $C_i^2$
 of $M_i$ containing $e_1$ and $e_2$ respectively 
 for distinct $e_1,e_2\in Z\setminus\{e\}$. Then $C'=C_1^1\bigtriangleup C_2^1$ is a cycle of $M$ contradicting the minimality of $C$. Hence, for each $i\in\{1,2\}$,  $C_i'$  is a circuit  of $M_i$.

In the opposite  direction, if $C$ is a circuit of $M_1$ or $M_2$ such that $C\cap Z=\emptyset$, then   $C$ is a circuit of $M$. Suppose now that $C=C_1\bigtriangleup C_2$, where $C_1$ and $C_2$ are circuits of $M_1$ and $M_2$ respectively, $C_1\cap Z=\{e\}$  and $C_2\cap Z=\{e\}$
 for some $e\in Z$,  and $C_1\bigtriangleup Z$  or $C_2\bigtriangleup Z$  is a circuit of $M_1$  or $M_2$  respectively. We show that $C$ is a circuit of $M$.

To obtain a contradiction, assume that $C$ is not a circuit. By Observation~\ref{obs:symm}, $C$ is a cycle of $M$. Therefore, there is a circuit $C'\subset C$. If $C'\subseteq E(M_1)$ or $C'\subseteq E(M_2)$, then $C'\subset C_1$ or $C'\subset C_2$, but this  contradicts   the condition that $C_1$ and $C_2$  are circuits of $M_1$ and $M_2$ respectively. Hence, $C'\setminus E(M_1)\neq \emptyset$ and $C'\setminus E(M_2)\neq\emptyset$. As we already proved above, $C'=C_1'\bigtriangleup C_2'$, where $C_i'$ is  a circuit of $M_i$, $i\in \{1,2\}$,  and $C_i'\cap Z=\{e'\}$ and $C_2'\cap Z=\{e'\}$ for some $e'\in Z$. Clearly, $C_1'\setminus\{e'\}\subseteq C_1\setminus\{e\}$ and $C_2'\setminus\{e'\}\subseteq C_2\setminus\{e\}$ and at least one of the inclusions is proper. If $e'=e$, then 
$C_1'\subseteq C_1$ and $C_2'\subseteq C_2$ and at least one of the inclusions is proper contradicting the fact that $C_1$ and $C_2$ are circuits of $M_1$ and $M_2$ respectively. Hence, $e'\neq e$. If $C_1'\setminus\{e'\}= C_1\setminus\{e\}$, then $\{e,e'\}=C_1'\bigtriangleup C_1$. This  contradicts the condition that $Z$ is a circuit. Hence, $C_1'\setminus\{e'\}\subset C_1\setminus\{e\}$. But then $C_1'\subset C_1\bigtriangleup Z$, and  therefore $C_1\bigtriangleup Z$ is not a circuit of $M_1$. Symmetrically,  $C_2\bigtriangleup Z$ is not a circuit of $M_2$; a contradiction. Hence, $C$ is a circuit of $M$.
\end{proof}

We conclude this section by the following lemma about circuits in graphic and cographic matroids.

\begin{lemma}\label{lem:gr-cogr}
Let $Z=\{e_1,e_2,e_3\}$ be a circuit of a binary matroid $M$. Let also $C$ be a circuit of $M$ such that $C\cap Z=\{e_3\}$. Then  the following holds:
\begin{itemize}
\item[i)] If $M=M(G)$ for a graph $G$, then  $C'=C\bigtriangleup Z$ is a circuit of $M$ if and only if
$C$ induces a cycle of $G-v$, where $v$ is the vertex of $G$ incident with $e_1$ and $e_2$.
\item[ii)] If $M=M^*(G)$ for a connected graph $G$, then  $C'=C\bigtriangleup Z$ is a circuit of $M$ if and only if
$C=E(A,B)$ for a cut $(A,B)$ of $G$ such that $G[A]$ and $G[B]$ are connected graphs and either $e_1,e_2\in E(G[A])$,  or $e_1,e_2\in E(G[B])$.   
\end{itemize}
\end{lemma}

\begin{proof}
The first claim is straightforward. To show ii), recall that $C$ is a minimal cut-set of $G$. Hence, there is a cut $(A,B)$ of $G$ such that $C=E(A,B)$ and $G[A]$ and $G[B]$ are connected. 

Assume that $e_1\in E(G[A])$ and $e_2\in E(G[B])$. Since $Z$ is a minimal cut-set of $G$, we have that $e_1$ and $e_2$ are bridges of $G[A]$ and $G[B]$ respectively. Then $C\bigtriangleup Z$ is a cut-set  separating $G$ into 3 components. Hence  $C'$ is not a minimal cut-set, which is  a contradiction. Therefore, either $e_1,e_2\in E(G[A])$,  or $e_1,e_2\in E(G[B])$. 

Suppose now that $C=E(A,B)$ for a cut $(A,B)$ of $G$ such that $G[A]$ and $G[B]$ are connected and
 $e_1,e_2\in E(G[A])$. Because $Z$ is a minimal cut-set, $\{e_1,e_2\}$ is a minimal cut-set of $G[A]$. Let $(A_1,A_2)$ be a cut of $G[A]$ such that $E(A_1,A_2)=\{e_1,e_2\}$. Assume that the end-vertex of $e_3$ in $A$ is in $A_1$. Since $Z$ is a minimal cut-set, the edges of $C\setminus\{e_3\}$ join $A_2$ with $B$. It implies, that $C\bigtriangleup Z$ is a minimal cut-set that separates $A_2$ and $A_1\cup B$.  
\end{proof}

\section{Minimal cut with specified edges}\label{sec:cut-graph}
To construct an algorithm for \wmsc{} for regular matroids, we need an algorithm for cographic matroids. 
Let $G$ be a connected graph, and let   $T\subseteq E(G)$ be a set of terminal edges. For  sets $R_1,R_2\subseteq V(G)$, we say that 
$C\subseteq E(G)$ is \emph{$(R_1,R_2)$-terminal cut-set} if $C$ is  (a) a minimal cut-set; (b) $C\supseteq T$;  and (c) $G-C$ contains   distinct connected components $X_1$ and $X_2$ such that $R_i\subseteq X_i$ for $i\in\{1,2\}$.

We will need solve  the following auxiliary parameterized problem  

\defparproblema{\emc}
{A connected graph $G$, a weight function $w\colon E(G)\rightarrow \mathbb{N}$, a set of terminals $T\subseteq E(G)$,  sets $R_1,R_2\subseteq V(G)$, and a positive integer $k$.}%
{$k$}
{Decide whether  $G$ contains an $(R_1,R_2)$-terminal cut-set $C$  such that $w(C)-w(T)\leq k$.}

\smallskip
We say that an {$(R_1,R_2)$-terminal cut-set} $C$ with the required weight is a \emph{solution}  of \emc. Observe that if in the instance of \emc we have
$R_1\cap R_2\neq\emptyset$, then the problem does not have a solution and this is  a no-instance.

In what follows, we prove that \emc is \classFPT{}.
In the special case when $R_1=R_2=\emptyset$, \emc{} essentially asks for a minimum weight minimal cut of a graph that contains specified edges. 
We believe that this graph problem is interesting in its own.

\begin{theorem}\label{thm:graph-cut}
\emc is solvable in time $2^{\Oh(k^2\log k)}\cdot n^{\Oh(1)}$.
\end{theorem}

The proof of Theorem~\ref{thm:graph-cut} is technical and  is given in the remaining part of the section.  
It is based on a (non-trivial) application of  the recent
algorithmic technique 
  of \emph{recursive understanding}  introduced by Chitnis et al. in~\cite{ChitnisCHPP12} (see also~\cite{ChitnisCHPP12a} for more details). 
  
\subsection{Preliminaries}\label{sec:rec-und}
First, we introduce some notions required for the proof of Theorem~\ref{thm:graph-cut}.

Let $G$ be a graph, $X\subseteq V(G)$. We say that $G'$ is obtained from $G$ by the \emph{contraction of $X$}, if we get $G'$ by deleting the vertices of $X$ and replacing by a vertex $x$, and then each edge $uv\in E(G)$ with $u,v\in X$ is replaced by a loop $xx$, and each edge $uv\in E(G)$ with $u\in X$ and $v\notin X$ is replaced by $xv$. Notice that while contracting, we do not reduce the number of edges, and that  we can obtain loops and multiple edges by this operation. For simplicity, we do not distinguish edges of the original graph from the edges obtained from them by contracting a set if it does not create confusions. 
 For an edge weighted graph, we assume that every new edge has the same weight as the edge it replaces.  
To simplify notations, throughout this section we also assume that if the contraction is done for some set $X\subseteq V(G)$ in an instance $(G,w,T,R_1,R_2,k)$ of \emc, 
 then if $X\cap R_i\neq\emptyset$ for $i\in\{1,2\}$, then the vertex obtained from $X$ is in $R_i$, and if a terminal edge is replaced, then the obtained edge is included in $T$. 

For a set $X$, we denote by $\mathcal{P}(X)$ the set of all partitions of $X$. We assume that  $\mathcal{P}(X)=\emptyset$ if $X=\emptyset$.

The main idea behind the recursive understanding technique~\cite{ChitnisCHPP12} is the following. We try to find a minimal cut-set of bounded size that separates an input graph into two sufficiently big parts.
If such a cut-set exists, then we solve the problem recursively for one of the parts and replace this part by an equivalent graph of bounded size; the equivalence here means that the replacement keeps all essential solutions of the original part. In our case, the replacement is obtained by contracting some edges. This way, we obtain a graph of smaller size. 
If the input graph has no cut-set with the required properties, then it either has a bounded size or has high connectivity. In the case of the bounded size graph we can apply brute force, and if the graph is highly connected, then we can exploit this property to solve the problem.  
To define formally what we mean by high connectivity, we need the following definition. 
\begin{definition}[\cite{ChitnisCHPP12}]
Let $G$ be a connected graph and let $p,q$ be positive integers. A cut $(A,B)$ of $G$ is called a \emph{$(q,p)$-good edge separation} if
\begin{itemize}
\item[i)] $|A|,|B|>q$,
\item[ii)] $|E(A,B)|\leq p$,
\item[iii)] $G[A]$ and $G[B]$ are connected.
\end{itemize}
\end{definition}

Let $G$ be a connected graph and let $p,q$ be positive integers. We say that $G$ is  $(q,p)$-unbreakable if there is no cut $(A,B)$ of $G$ such that 
\begin{itemize}
\item[$i)$] $|A|,|B|>q$, and 
\item[$ii)$] $|E(A,B)|\leq p$,
\end{itemize}

Chitnis et al. proved the following lemma~\cite{ChitnisCHPP12}.

\begin{lemma}[\cite{ChitnisCHPP12}]\label{lem:good}
There  exists  a  deterministic  algorithm that,  given a  connected  graph $G$ along  with  integers $p$ and $q$,  in  time
$2^{\Oh(\min\{p,q\}\log (p+q))}\cdot n^3\log n$
either finds a $(q,p)$-good edge separation, or correctly concludes that no such separation exists.
\end{lemma}

We use this lemma to show the following.

\begin{lemma}\label{lem:good-unbreak}
There  exists  a  deterministic  algorithm that,  given a  connected  graph $G$ along  with  integers $p$ and $q$,  in  time
$2^{\Oh(\min\{p,q\}\log (p+q))}\cdot n^3\log n$
either finds a $(q,p)$-good edge separation, or correctly concludes that $G$ is $(pq,p)$-unbreakable.
\end{lemma}

\begin{proof}
We use Lemma~\ref{lem:good} to find a $(q,p)$-good edge separation. If the algorithm returns a $(q,p)$-good edge separation, we return it. Assume that the algorithm reported that no such separation exists. We claim that $G$ is $(pq,p)$-unbreakable. To obtain a contradiction, assume that $(A,B)$ is a cut of $G$ such that $|A|,|B|>pq$ and $|E(A,B)|\leq p$. Consider $G[A]$. Because $G$ is connected and $|E(A,B)|\leq p$, $G[A]$ has at most $p$ components. Hence, $G$ has a component $H_A$ with at least $q+1$ vertices. Symmetrically, we obtain  that $G[B]$ has a components $H_B$ with at least $q+1$ vertices. Let $C$ be a minimum cut-set in $G$ that separates $V(H_A)$ and $V(H_B)$. Clearly, $|C|\leq p$.
Let $(A',B')$ be the cut of $G$ with $V(H_A)\subseteq A'$, $V(H_B)\subseteq B'$ and $E(A',B')=C$. We have that $(A',B')$ is a $(q,p)$-good separation, but it contradicts the assumption that the algorithm reported that there is no such a separation.
\end{proof}

We use Lemma~\ref{lem:good-unbreak} to find a $(q,p)$-good edge separation for appropriate $p$ and $q$. If such a cut $(A,B)$ exists, we solve the problem recursively for one of the parts, say, for $G[A]$. But to be able to obtain a solution for the original instance, we should combine solutions for the both parts. We use the fact that $G[A]$ is separated from the remaining part of the graph by a small number of vertices that are the end-vertices of the edges of the cut-set which are called \emph{border terminals}. (In fact, we keep $2p$ border terminals to execute the recursive step.) As we want to find all essential solutions for $G[A]$ to replace this graph by a graph of bounded size, we have to take into account all possibilities for the part of a solution in $B$ to separate the border terminals.

This leads us to the following definition. 
Let $({G},w,T,{R}_1,{R}_2,{k})$ be an instance of \emc given together with a set   $X\subseteq V(G)$  of  border terminals of $G$.  
We say that an instance $(\hat{G},w,T,\hat{R}_1,\hat{R}_2,\hat{k})$ of \emc is obtained from $({G},w,T,{R}_1,{R}_2,{k})$ by \emph{border contraction} if $\hat{k}\leq k$ and  there is 
a partition $(X_1,\ldots,X_t)\in\mathcal{P}(X)$ and partition $(I_1,I_2)$ of $\{1,\ldots,t\}$, where $I_i$ can be empty, such that 
 $\hat{G}$ is obtained by consecutively  contracting $X_1,\ldots,X_t$, and setting   $\hat{R}_i=R_i\cup\{x_j\mid  j\in I_i\}$ for $i\in\{1,2\}$, where each $x_j$ is the vertex obtained from $X_j$ by contraction. Let us note that the total number of different border contractions  of a given instance depends only on the size of $X$ and $k$ and is  $k\cdot |X|^{\cO(|X|)}$.

It leads us to the following auxiliary problem. In this problem we have to output a solution (if there is any) for each of the instances 
of \emc obtained by all possible border contractions of a given instance. 
Notice that this is not a decision problem.

\defparproblema{\bemc{}}%
{A connected graph $G$, a weight function $w\colon E(G)\rightarrow \mathbb{N}$, a set of terminals $T\subseteq E(G)$,  sets $R_1,R_2\subseteq V(G)$, a positive integer $k$, and 
a sets of border terminals $X\subseteq V(G)$ with $|X|\leq 4k$.}%
{k}%
{Output for each possible  instance of \emc which can be obtained from $({G},w,T,{R}_1,{R}_2,{k})$ by border contractions of  $X$ a solution, if there is any. 
In a case when a border contraction instance has no solution, output $\emptyset$. 
 }

\smallskip
Thus an output for   \bemc is a family of edge sets, where the total number of edges in the solution is at most 
$k\cdot(4k)^{4k}\cdot 2^{4k}=2^{\Oh(k\log k)}$. 
Notice also that to solve \emc{}, we can apply an algorithm for \bemc{} for the special case $X=\emptyset$. 

\subsection{High connectivity phase}\label{sec:conn}
In this section we construct an algorithm for \bemc{} for the case when an input graph is  $(pq,p)$-unbreakable for $p=2k$ and 
$q=k^2\cdot2^{4k+4k\log 4k}+4k+1$; 
we fix the values of $p$ and $q$ for the remaining part of Section~\ref{sec:cut-graph}. First, we solve  \emc{} and then explain how to obtain the algorithm for \bemc{}.

\begin{lemma}\label{lem:OCT}
Let $G$ be a graph with an edge weight function  $w\colon E(G)\rightarrow \mathbb{N}$, $T\subseteq E(G)$
and let $k$ be a positive integer. It can be decided in time $2^{\cO(k)}\cdot n^{\Oh(1)}$ whether there is a 
cut $(A,B)$ of $G$ such that $T\subseteq E(A,B)$, and $w(E(A,B)\setminus T)\leq k$.
\end{lemma}

\begin{proof}
We show the lemma by the reduction of the problem to the \textsc{Odd Cycle Transversal (OCT)} problem. Let us remind that in the OCT problem we are given  a graph $G$ and a positive integer $k$, the task is to decide whether  there is a set of at most $k$ vertices $S$ such that $G-S$ is bipartite. Since OCT is known to be solvable in time $2^{\cO(k)}\cdot n^{\Oh(1)}$, this will prove the lemma.

Let $G$ be a graph with an edge weight function  $w\colon E(G)\rightarrow \mathbb{N}$, $T\subseteq E(G)$,  and let $k$ be a positive integer. 
Recall that we allow loops and multiple edges. To slightly simplify reduction, we first exhaustively apply two simple reduction rules.

If $e\in T$ is a loop, then  
$e\notin E(A,B)$ for any cut $(A,B)$.  If  a loop $e\notin T$, then $e$ is irrelevant. Hence we have the following reduction rule.

\begin{reduction}[{\bf Loop reduction rule}]\label{rule:loop}Let $e\in E(G)$ be a loop. 
If  $e\not\in T$, then  delete $e$. Otherwise (if $e\in T$) report that there is no required cut $(A,B)$. 
 \end{reduction}

Clearly, any two parallel edges are either both included in a cut-set or both are excluded from it. Notice also that the weights of terminals are irrelevant. Hence, we can safely apply the following rule. 

\begin{reduction}[{\bf Parallel terminal reduction rule}]\label{rule:par}
If there are two parallel  edges $e_1,e_2\in T$, delete one of them and change the weight of the remaining edge to $1$. 
\end{reduction}

From now on we assume that the rules cannot be applied.
We construct (unweighted)  graph $G'$ from $G$ as follows.
\begin{itemize}
\item Subdivide each edge $uv\notin T$, that is, add a new vertex $z_{uv}$ and replace $uv$ by $uz_{uv}$ and $vz_{uv}$; we call the new vertices \emph{subdivision} vertices.
\item Replace each subdivision vertex $z_{uv}$ by $r=\min\{w(uv),k+1\}$ false twins, i.e., we replace $z_{uv}$ by $r$ vertices adjacent to $u$ and $v$; denote by $Z_{uv}$ the set of obtained vertices.
\item Replace each vertex $v$ of $V(G)$ by $k+1$ false twins, i.e., we replace $v$ by $k+1$ vertices with the same neighbors as $v$; denote by $U_v$ the set of obtained vertices.
\end{itemize}
Notice that because of reduction rules, $G'$ is a simple graph.
We claim that there is a cut $(A,B)$ of $G$ such that $T\subseteq E(A,B)$, and $w(E(A,B)\setminus T)\leq k$
 if and only of $(G',k)$ is a yes-instance of \textsc{OCT}.

Suppose that $(A,B)$ is a cut of $G$  such that $T\subseteq E(A,B)$, and $w(E(A,B)\setminus T)\leq k$. 
We construct the set $S\subseteq V(G')$ by including in $S$ the set of vertices $Z_{uv}$ for each $uv\in E(A,B)\setminus T$. Then $G'-S$ is   bipartite.

Suppose  that there is $S\subseteq V(G')$ of size at most $k$ such that $G'-S$ is bipartite. Without loss of generality we assume that $S$ is an inclusion minimal set with this property.  Because $S$ is minimal, 
  if $x$ and $y$ are false twins of $G$, then either $x,y\in S$, or $x,y\notin S$. Let $(X,Y)$ be a bipartition of $G'-S$. 
Since $|U_v|>k$, we have that $U_v\cap S=\emptyset$ for $v\in V(G)$. Notice also that we can assume that either $U_v\subseteq X$ or $U_v\subseteq Y$ for $v\in V(G)$, as otherwise, if there is $v\in V(G)$ such that $U_v\cap X\neq\emptyset$ and $U_v\cap Y\neq\emptyset$, then the vertices of $U_v$ are isolated vertices of $G'-S$. 
Let $A=\{v\in V(G)\mid U_v\subseteq X\}$ and  $B=\{v\in V(G)\mid U_v\subseteq Y\}$. Clearly, $(A,B)$ is a cut of $G$. Let $uv\in T$. Assume that $U_u\subseteq X$. Then $U_v\subseteq Y$ and, therefore, $uv\in E(A,B)$. Let $uv\in E(A,B)\setminus T$ and assume that $u\in A$ and $v\in B$. Then $U_u\subseteq X$ and $U_v\subseteq Y$. Hence, $Z_{uv}\subseteq S$. Since $|Z_{uv}|=\min\{w(uv),k+1\}$ and $|S|\leq k$, the total weight of the edges of $E(A,B)\setminus T$ is at most $k$.

This proves the correctness of the reduction. Since \textsc{OCT} can be solved in time $2.3146^k\cdot n^{\Oh(1)}$ by the results of Lokshtanov et al.~\cite{LokshtanovNRRS14}, we get the claim of the lemma.
\end{proof}

Let $(A_1,B_1)$ and $(A_2,B_2)$ be cuts of a graph $G$. We define  the \emph{distance} between these cutsas 
$$\dist((A_1,B_1),(A_2,B_2))=\min\{|A_1\bigtriangleup A_2|, |A_1\bigtriangleup B_2|\}.$$

The following structural lemmata are crucial for our algorithm.

\begin{lemma}\label{lem:two-cuts}
Let $G$ be a graph with an edge weight function  $w\colon E(G)\rightarrow \mathbb{N}$, set of terminals $T\subseteq E(G)$, and let $k$ be a positive integer. Let  $(A_1,B_1)$ and $(A_2,B_2)$ be cuts of $G$ such that $T\subseteq E(A_i,B_i)$ and $w(E(A_i,B_i)\setminus T)\leq k$ for $i\in\{1,2\}$. Then $w(E(A_1\bigtriangleup A_2,A_1\bigtriangleup B_2))\leq 2k$.
\end{lemma}

\begin{proof}
Notice that $(A_1\bigtriangleup A_2,A_1\bigtriangleup B_2)$ is a cut of $G$. 
For each $i\in\{1,2\}$, we have that  $T\subseteq E(A_i,B_i)$. Therefore, 
the set 
\[E(A_1\cap A_2,A_1\cap B_2)\cup E(B_1\cap A_2,B_1\cap B_2) \cup E(A_2\cap A_1,A_2\cap B_1)\cup E(B_2\cap A_1,B_2\cap B_1) \] 
does not contain edges from $T$.
 
Hence,  $$E(A_1\cap A_2,A_1\cap B_2)\cup E(A_2\cap B_1 ,B_1\cap B_2)\subseteq E(A_2,B_2)\setminus T,$$ and therefore, 
$$w(E(A_1\cap A_2,A_1\cap B_2)\cup E(A_2\cap B_1,B_1\cap B_2)\leq k.$$ Symmetrically, $$w(E(A_1\cap A_2,A_2\cap B_1)\cup E(A_1\cap B_2,B_1\cap B_2))\leq k.$$
Since
$$A_1\bigtriangleup A_2=(A_1\cap B_2)\cup (A_2\cap B_1) \text{ and }
 A_1\bigtriangleup B_2=(A_1\cap A_2)\cup (B_1\cap B_2),$$
the claim follows.
\end{proof}

 Let us recall that in this section we fix $p=2k$ and $q=k2^{4k+4k\log 4k}+4k+1$. 
\begin{lemma}\label{lem:struct-cuts}
Let $G$ be a connected $(pq,p)$-unbreakable graph with an edge weight function  $w\colon E(G)\rightarrow \mathbb{N}$, $T\subseteq E(G)$ and let $k$ be a positive integer. Let  $(A_1,B_1)$ and $(A_2,B_2)$ be cuts of $G$ such that $T\subseteq E(A_i,B_i)$ and $w(E(A_i,B_i)\setminus T)\leq k$ for $i\in\{1,2\}$. Then  \[\dist((A_1,B_1),(A_2,B_2))\leq pq.\]
\end{lemma}

\begin{proof}Aiming towards a 
 a contradiction, we assume that $\dist((A_1,B_1),(A_2,B_2))> pq$.  Let us note  that  $(A_1\bigtriangleup A_2,A_1\bigtriangleup B_2)$ is a partition of $V(G)$. 
Since $\dist((A_1,B_1),(A_2,B_2))> pq$,  we have that $|A_1\bigtriangleup A_2|>pq$ and $|A_1\bigtriangleup B_2|>pq$.
By Lemma~\ref{lem:two-cuts}, $ w(E(A_1\bigtriangleup A_2,A_1\bigtriangleup B_2))\geq |A_1\bigtriangleup A_2,A_1\bigtriangleup B_2)| \leq 2k=p$;  contradicting the assumption that $G$ is  $(pq,p)$-unbreakable.
\end{proof}

Our algorithm for \emc{}  uses the \emph{random separation} technique proposed by Cai, Chan and Chan~\cite{CaiCC06}.  For derandomization, we use the following lemma proved by Chitnis et al.~\cite{ChitnisCHPP12}.

\begin{lemma}[\cite{ChitnisCHPP12}]\label{lem:derand}
Given a set $U$ of size $n$, and integers $0\leq a,b\leq n$, one can in time $2^{\Oh(\min\{a,b\}\log (a+b))}\cdot n\log n$ construct a family $\mathcal{F}$ of at most 
 $2^{\Oh(\min\{a,b\}\log (a+b))}\cdot\log n$ subsets of $U$, such that the following holds: for any sets $A,B\subseteq U$, $A\cap B=\emptyset$, $|A|\leq a$, $|B|\leq b$, there exists a set $S\in\mathcal{F}$ with $A\subseteq S$ and $B\cap S=\emptyset$.
\end{lemma}

Now we are ready to give the algorithm for \emc{} for unbreakable graphs.

\begin{lemma}\label{thm:graph-cut-high}
\emc{} can be solved in $2^{\Oh(k^2\log k)}n^{\Oh(1)}$ time for $(pq,p)$-unbreakable graphs.
\end{lemma}

\begin{proof}
Let $(G,w,T,R_1,R_2,k)$ be an instance of \emc{}, where $G$ is  $(pq,p)$-unbreakable. If $n\leq pq$, we solve the problem by the brute force selection of at most $k$ edges in time $2^{\Oh(k^2\log k)}n^{\Oh(1)}$. From now we assume that $n>pq$.

Using Lemma~\ref{lem:OCT}, we find a cut $(A,B)$ of $G$ such that $T\subseteq E(A,B)$ and $w(E(A,B)\setminus T)\leq k$. If such a cut does not exist, we conclude that we are a given a no-instance.

Let $(G,w,T,R_1,R_2,k)$ be a yes-instance and let $C=E(A',B')$ be a solution. Without loss of generality, we  assume that $\dist((A,B),(A',B'))=|A\bigtriangleup A'|$. By Lemma~\ref{lem:struct-cuts}, 
$|A\bigtriangleup A'|\leq pq$.  It means, that to solve the problem, we can either  find a cut $(A',B')$ or, equivalently, $A\bigtriangleup A'$ with these properties or conclude correctly that such a cut does not exist.
First, we describe a randomized algorithm which finds $A\bigtriangleup A'$ and then explain how to derandomize it.

We randomly color the vertices of $V(G)\setminus R$ by two colors \emph{red} and \emph{blue} with the probabilities $1-\frac{1}{pq}$ and $\frac{1}{pq}$ respectively. 
We are looking for a set $X\subseteq V(G)$  such that the following holds:
\begin{itemize}
\item[i)]  $|X|\leq pq$.
\item[ii)] For $A'=A\bigtriangleup X$ and $B'=V(G)\setminus A'$, $C=E(A',B')$ is a solution for $(G,w,T,R_1,R_2,k)$.
\item[iii)] The vertices of $X$ are red and the vertices of $N_G(X)$ are blue.
\end{itemize}
We say that $C=E(A',B')$ is a \emph{colorful solution}.

The vertices of $G$ are colored in red and in blue induce subgraphs that we call \emph{red} and \emph{blue} correspondingly. We also say that $H$ is a \emph{red component}  if $H$ is a connected component of the red (respectively, blue) subgraph of $G$. Because of i)--iii), we have the following properties:
\begin{itemize}
\item if $H$ is a red component, then either $V(H)\subseteq X$ or $V(H)\cap X=\emptyset$,
\item if $v\in V(G)$ is colored blue, then $v\notin X$.
\end{itemize}
We use i)--iii) and these properties to obtain reduction rules that recolor red components in  blue, that is, each vertex of such a component becomes blue. We apply these rules exhaustively.

Since $T\subseteq E(A',B')$ if $C=E(A',B')$ is a solution for $(G,w,T,R_1,R_2,k)$, we get the following rule.

\begin{reduction}[{\bf $T$-reduction rule}]\label{rule:T-red}
If there is $uv\in T$ such that $u$ is red and $v$ is blue, then recolor the red component $H$ containing $u$ in blue.
\end{reduction}

We say that $uv\in E(G)$ is a \emph{crossing edge for a red component $H$} if $u\in V(H)$, $v\notin V(H)$, and either $u\in A$ and $v\in B$ or $u\in B$ and $v\in A$. Notice that $v$ is colored blue.  
Notice also that if $H$ is a red component without crossing edges and $V(H)\subseteq X$, then for $A'=A\bigtriangleup X$ and $B'=V(G)\setminus A'$,
$V(H)\cap A'$ and $V(H)\cap B'$ induce components of $G[A']$ and $G[B']$ respectively. If $|V(H)|\leq pq$, then we have that $G[A']$ or $G[B']$ is not connected, because $|V(G)|>pq$. Hence, $V(H)\cap X=\emptyset$ if $C=E(A',B')$ is a solution for $(G,w,T,R,k)$.
It gives us the next rule.

\begin{reduction}[{\bf Crossing reduction rule}]\label{rule:cross-red}
If there is a red component $H$ without crossing edges, then recolor $H$ blue.
\end{reduction}
 
After the exhaustive applications of Rules~\ref{rule:T-red} and \ref{rule:cross-red}, each red component $H$ has crossing edges and these crossing edges are not in $T$. Since $w(E(A,B)\setminus C)\leq k$, the total number of crossing edges is at most $k$ and, therefore, there are at most $k$ red components. Because $X$ is a union of some red components, we check all possibilities for $X$ (the number of all possibilities is  at most $2^{k}$), and for each choice, we check whether $C=E(A',B')$ is a solution for $(G,w,T,R_1,R_2,k)$. 
If we do not succeed in finding a solution for at least one 
 of the choices, then we return that there is no solution.

Since Rules~\ref{rule:T-red} and \ref{rule:cross-red} can be run  in polynomial time, a colorful solution for $(G,w,T,R_1,R_2,k)$ can be found in time $2^k\cdot n^{\Oh(1)}$. 

Our next aim is to evaluate the probability of existence of a colorful solution for  $(G,w,T,R_1,R_2,k)$ if  $(G,w,T,R_1,R_2,k)$ is a yes-instance of  \emc{}. 
Assume that $(G,w,T,R_1,R_2,k)$ is a yes-instance and $C=E(A',B')$ is a solution, where $(A',B')$ is a cut of $G$. We assume that $\dist((A,B),(A',B'))=|A\bigtriangleup A'|$. Let $X=A\bigtriangleup A'$.
By Lemma~\ref{lem:struct-cuts},  $|X|=\dist((A,B),(A',B'))\leq pq$. By Lemma~\ref{lem:two-cuts}, $|E(X,V(G)\setminus X)|\leq 2k$ and, therefore, $|N_G(X)|\leq 2k$. Then the probability that the vertices of $X$ are colored red and the vertices of $N_G(X)$ are colored blue is at least $(1-\frac{1}{pq})^{pq}\cdot \frac{1}{(pq)^{2k}}\geq \frac{1}{4(pq)^{2k}}$ if $pq\geq 2$.
If we run our randomized algorithm $N=4(pq)^{2k}$ times, then the probability that we do not have a colorful solution for each of the $N$ random colorings, is at most 
$(1-\frac{1}{4(pq)^{2k}})^N\leq e^{-1}$. It means, that it is sufficient to run the algorithm $N$ times to claim that if we do not find a solution for $N$ random colorings, then with   probability at least $1-e^{-1}>0$, $(G,w,T,R_1,R_2,k)$ is a no-instance. In other words,  we have a true-biased Monte-Carlo algorithm which  runs in time $N\cdot 2^k\cdot n^{\Oh(1)}$ if the initial partition $(A,B)$ is given. Since  $p=2k$ and $q=k2^{4k+4k\log 4k}+4k+1$ and the initial partition $(A,B)$ can be found in time $2^{\Oh(k)}\cdot n^{\Oh(1)}$, see  Lemma~\ref{lem:OCT}, the total running time is $2^{\Oh(k^2\log k)}\cdot n^{\Oh(1)}$.

To derandomize the algorithm, we use Lemma~\ref{lem:derand} for $a=q$, $b=p$ and $U=V(G)$.  We construct the family $\mathcal{F}$ of subsets of $V(G)$ described in Lemma~\ref{lem:derand}, and instead of random colorings, for each $S\in \mathcal{F}$,  we consider the coloring of $G$ such that the vertices of $S$ are colored red and the vertices of $V(G)\setminus S$ are blue.
Lemma~\ref{lem:derand} guarantees that $(G,w,T,R_1,R_2,k)$ is a yes-instance of  \emc{} if and only if we have a colorful solution for at least one of $|\mathcal{F}|$ colorings. Since $\mathcal{F}$ can be constructed in time $2^{\Oh(k^2\log k)}\cdot n^{\Oh(1)}$ and $|\mathcal{F}|=2^{\Oh(k^2\log k)}\cdot n^{\Oh(1)}$, the running time of the derandomized algorithm is $2^{\Oh(k^2\log k)}\cdot n^{\Oh(1)}$.
\end{proof}

We use Lemma~\ref{thm:graph-cut-high}, to solve \bemc{}.

\begin{lemma}\label{lem:graph-cut-bord-high}
\bemc{} can be solved in time $2^{\Oh(k^2\log k)}\cdot n^{\Oh(1)}$ for $(pq,p)$-unbreakable graphs.
\end{lemma}

\begin{proof}
Let $(G,w,T,R_1,R_2,k,X)$ be an instance of \bemc. Let us recall that the output of \bemc{} consists of solutions of 
\emc for all possible 
border contractions      $(\hat{G},w,T,\hat{R}_1,\hat{R}_2,\hat{k})$ of  $(G,w,T,R_1,R_2,k,X)$.
 Notice that if $G$ is $(pq,p)$-unbreakable, then each graph $\hat{G}$ is $(pq,p)$-unbreakable as well, because contractions of sets do not violate this property. We  apply Lemma~\ref{thm:graph-cut-high} for each instance $(\hat{G},w,T,\hat{R},\hat{k})$ of \emc. Since  the number of all possible border contractions is in  $2^{\Oh(k\log k)}$, the  
 total running time required to output the family of edge sets for \bemc
 is  in  $2^{\Oh(k^2\log k)}\cdot n^{\Oh(1)}$.
\end{proof}

\subsection{Proof of Theorem~\ref{thm:graph-cut}}\label{sec:rec}
We are ready to proceed with the proof of Theorem~\ref{thm:graph-cut}, which says that 
\emc is solvable in time $2^{\Oh(k^2\log k)}\cdot n^{\Oh(1)}$. 
We  give a  recursive algorithm  solving the more general \bemc. Then to solve \emc we solve the  special case $X=\emptyset$
of \bemc. 
 Recall that we fixed $p=2k$ and $q=k2^{4k+4k\log 4k}+4k+1$.

Let $(G,w,T,R_1,R_2,k,X)$ be an instance of \bemc.

It is convenient to sort out a trivial case first. Notice that if for $(G,w,T,R_1,R_2,k,X)$  the set of terminal edges is a cut-set of the input graph but not a minimal cut-set, then this is a no-instance. It gives us the following rule.

\begin{reduction}[{\bf Stopping rule}]\label{rule:trivial}
If  
 graph $G-T$ has at least two components without border terminals, then output the empty set for every partition 
$ (X_1,\ldots,X_t) $ and every partition $(I_1,I_2)$ of $\{1,\ldots,t\}$. 
\end{reduction}

From now we assume that Stopping rule is not applicable to the given instance. 

We apply Lemma~\ref{lem:good-unbreak} on $G$. If   $G$ is $(pq,p)$-unbreakable, then we apply Lemma~\ref{lem:graph-cut-bord-high} to solve the problem. Otherwise, 
 the algorithm from Lemma~\ref{lem:good-unbreak} returns a $(q,p)$-good edge separation $(U,W)$ of $G$. 

The set of border terminals $X$ has size at most $4k=2p$.  Hence, $|X\cap U|\leq p$ or $|C\cap W|\leq p$. Assume without loss of generality that $|X\cap U|\leq p$. 
Let $T'=T\cap E(G[U])$,
$R_1'=R_1\cap U$, $R_2'=R_2\cap U$, and denote by $w'$ the restriction of $w$ on $E(G[U])$. 
We also define the set of border terminals 
$$X'=(X\cap U)\cup\{v\in V(G[U])\mid v\text{ is incident with an edge of }E(U,W)\};$$ 
observe that $|X'|\leq 2p=4k$, because $|E(U,W)|\leq p$. We consider the instance \linebreak $(G[U],w',T',R_1',R_2',k,X')$ of \bemc{} and solve the problem recursively. 

Recall that the output of \bemc{}  for $(G[U],w',T',R_1',R_2',k,X')$ is a family of solutions $C$ for 
all possible border contractions. In other words, this is a family of solutions for 
  instances $(\hat{G}',w',T',\hat{R}_1',\hat{R}_2',\hat{k})$ 
for all $\hat{k}\leq k$ such that a solution exist, and  $\emptyset$ if there is no solutions.  Each $\hat{G}'$ and $\hat{R}_i'$ is constructed as follows:
for every partition $(X_1,\ldots,X_t)\in\mathcal{P}(X')$ and every partition $(I_1,I_2)$ of $\{1,\ldots,t\}$, where $I_i$ can be empty, we 
construct $\hat{G}'$ by consecutively  contracting $X_1,\ldots,X_t$, and set  $\hat{R}_i'=R_i'\cup\{x_j\mid  j\in I_i\}$ for $i\in\{1,2\}$, where each $x_j$ is the vertex obtained from $X_j$ by the contraction. For each of the subproblems,  solution $C$  is a set of edges of $G[U]$. 

Denote by $L$ the union of all sets generated by the algorithm for $(G[U],w',T',R_1',R_2',k,X')$. Let $G''$ be the graph obtained from $G$ by  contracting the edges of $E(G[U])\setminus (L\cup T)$. Denote by $\alpha\colon V(G)\rightarrow V(G'')$ the mapping that maps each vertex $v\in V(G)$ to the vertex obtained from $v$ by edge contractions. Let 
$R_1''=\alpha(R_1)$, $R_2''=\alpha(R_2)$ and $X''=\alpha(X)$. Notice that the edges of $T$ are not contracted. Denote by $T''$ the edges of $G''$ obtained from $T$; clearly, for each $uv\in T$, we have $\alpha(u)\alpha(v)\in T''$. For every $uv\in E(G)$ that was not contracted, the weight of the obtained edge $\alpha(u)\alpha(v)$ is $w''(\alpha(u)\alpha(v))=w(uv)$. We obtain a new instance $(G'',w'',T'',R_1'',R_2'',k,X'')$ of \bemc{}. As before, we do not distinguish between the edges obtained by contracting edges or the original edges; thus $T''=T$.

We claim that the original   $(G,w,T,R_1,R_2,k,X)$ and new  $(G'',w'',T'',R_1'',R_2'',k,X'')$ instances are equivalent in the following sense:  There is a solution (in fact, every solution) for  $(G'',w'',T'',R_1'',R_2'',k,X'')$ that is a solution for  $(G,w,T,R_1,R_2,k,X)$, and there is a solution for $(G,w,T,R_1,R_2,k,X)$ that is a solution for $(G'',w'',T'',R_1'',R_2'',k,X'')$.

\begin{lemma}\label{lem:eq}
For every partition $(X_1,\ldots,X_t)\in\mathcal{P}(X)$, every partition $(I_1,I_2)$ of $\{1,\ldots,t\}$, and every nonnegative $\hat{k}\leq k$, the instances 
$(\hat{G},w,T,\hat{R}_1,\hat{R}_2,\hat{k})$ and $(\hat{G}'',w'',T'',\hat{R}_1'',\hat{R}_2'',\hat{k})$ of \emc{} are equivalent, where $\hat{G}$ is constructed from $G$ by consecutive contracting $X_1,\ldots,X_t$, $\hat{R}_i=R_i\cup\{x_j\mid  j\in I_i\}$ for $i\in\{1,2\}$, where each $x_j$ is the vertex obtained from $X_j$ by  contraction, and, respectively,
$\hat{G}''$ is constructed from $G''$ by consecutive contracting $\alpha(X_1),\ldots,\alpha(X_t)$, $\hat{R}_i''=R_i''\cup\{x_j\mid  j\in I_i\}$ for $i\in\{1,2\}$, where each $x_j$ is the vertex obtained from $\alpha(X_j)$ by contraction.
\end{lemma}

\begin{proof}
Let  $P=(X_1,\ldots,X_t)\in\mathcal{P}(X)$, $(I_1,I_2)$ be a partition of $\{1,\ldots,t\}$ and $\hat{k}\leq k$.

Suppose that $(\hat{G},w,T,\hat{R}_1,\hat{R}_2,\hat{k})$ is a yes-instance and denote by $C$ a corresponding solution. Denote by $(A,B)$ the cut of $\hat{G}$ such that $C=E(A,B)$ and assume that $\hat{R}_1\subseteq A$ and $\hat{R}_2\subseteq B$.
Let $C'=C\cap E(G[U])$ and $k'=w(C'\setminus T)$.

We construct the partition $P'\in\mathcal{P}(X')$ in two stages. Recall that some of the border terminals in instance  $(G[U],w',T',R_1',R_2',k,X')$ could be also border terminals in the original instance. We include two such border terminals in the same set of $P'$ if they are in the same set of $P$. This way we obtain  partition $(Y_1,\ldots,Y_s)$ of $X'$. Then we replace two distinct sets $Y_i$ and $Y_j$, $i,j\in\{1,\ldots,s\}$,  by their union if
they can be ``connected" in $\hat{G}$ by a path avoiding  $G[U]$ and $C$. More precisely, if 
there are vertices $x\in Y_i$ and $y\in Y_j$ such that $\hat{G}$ contains an $(x',y')$-path,   where $x'$ and $y'$ are the vertices of $\hat{G}$ that are $x$ or $y$, or
obtained by contracting set containing $x$ or $y$ respectively, such that this path 
does not contain  edges of $G[U]$ and $C$.
 Notice that for any pair of such vertices $x$ an $y$, either $x',y'\in A$ or $x',y'\in B$, i.e., we never contract two vertices from different parts of the cut $(A,B)$. Denote by $(X_1'\ldots,X_r')$ the obtained partition $P'$ of $X'$. We define the partition $(I_1',I_2')$ of $\{1,\ldots,r\}$ by including $j\in\{1,\ldots,r\}$ in $I_1$ if $X_j'$ is obtained by contracting vertices of $A$ and we put $j$ in $I_2$ otherwise. Consider the instance 
 $(\hat{G}',w',T',\hat{R}_1',\hat{R}_2',\hat{k})$ of \emc{}, where $\hat{G}'$ 
is constructed form $G[U]$ by contracting $X_1',\ldots,X_r'$, and $\hat{R}_i'=R_i'\cup\{x_j\mid  j\in I_i'\}$ for $i\in\{1,2\}$, where each $x_j$ is the vertex obtained from $X_j'$ by  contraction. 

By the construction of $P'$ and $(I_1',I_2')$, we have that  $(\hat{G}',w',T',\hat{R}_1',\hat{R}_2',\hat{k})$ is a yes-instance. Hence, for instance $(G[U],w',T',R_1',R_2',k,X')$ the output of  \bemc 
 contains a solution $C''$ for this choice of  $P'$ and $(I_1',I_2')$, and $w(C''\setminus T)\leq k'$. Again, by the construction, we have that $S=(C\setminus C')\cup C''$ is a solution for $(\hat{G}'',w'',T'',\hat{R}_1'',\hat{R}_2'',\hat{k})$. Hence  $(\hat{G}'',w'',T'',\hat{R}_1'',\hat{R}_2'',\hat{k})$ is a yes-instance of \emc{}.

Finally,  if $(\hat{G}'',w'',T'',\hat{R}_1'',\hat{R}_2'',\hat{k})$ is a yes-instance, then \linebreak $(\hat{G},w,T,\hat{R}_1,\hat{R}_2,\hat{k})$ is a yes-instance, because $G''$ is obtained from $G$ by contracting nonterminal edges, and every solution for $(\hat{G}'',w'',T'',\hat{R}_1'',\hat{R}_2'',\hat{k})$ is a solution for $(\hat{G},w,T,\hat{R}_1,\hat{R}_2,\hat{k})$.
\end{proof}

By  Lemma~\ref{lem:eq}, that instead of deciding whether instance  $(G,w,T,R_1,R_2,k,X)$ is a yes-instance of \bemc, we can solve the problem on instance 
 $(G'',w'',T'',R_1'',R_2'',k,X'')$. What remains is to bound the size of $G''$, and this is what the next lemma does. 

\begin{lemma}\label{lem:red-size}
$|V(G'')|<|V(G)|$.
\end{lemma}

\begin{proof}
Recall that $G''$ is the graph obtained from $G$ by  contracting  the edges of $E(G[U])\setminus (L\cup T)$, where $L$ is the union of all sets generated by the algorithm for \bemc{} for $(G[U],w',T',R_1',R_2',k,X')$. Notice that for any $C$ in a solution for $(G[U],w',T',R_1',R_2',k,X')$, $w(C\setminus T)\leq k$. Hence, $|C\setminus T|\leq k$. Since $|X'|\leq 4k$, the total number of sets in the output is at most 
$k\cdot2^{4t}\cdot(4t)^{4t}$. 
Therefore, the graph $H$ obtained from $G[U]$ by contracting the edges of $E(G[U])\setminus (L\cup T)$ has at most 
$k^2\cdot 2^{4t}\cdot(4t)^{4t}$ 
nonterminal edges. Notice that $G[U]-T$ has at most $4k+1$ components, because of Rule~\ref{rule:trivial}. Hence,  $H$ has at most 
$k^2\cdot 2^{4t}\cdot(4t)^{4t}+4k+1\leq q$ 
vertices. Since $(U,W)$ is a $(q,p)$-good edge separation, $|V(H)|<|U|$. As we replace $G[U]$ by $H$ to construct $G''$, the claim follows.
\end{proof}

Lemma~\ref{lem:red-size} shows that we reduce the size of an input graph at each iterative step. 
Together with  Lemma~\ref{lem:graph-cut-bord-high}, it implies that  \bemc{}  is solvable in time $2^{\Oh(k^2\log k)}\cdot n^{\Oh(1)}$. This concludes the proof of Theorem~\ref{thm:graph-cut}.

\section{Solving \wmsc{} on regular matroids}\label{sec:wmsc}
This section is devoted to the proof of the first main result of the paper.  

\begin{theorem}\label{thm:wmsc} 
\wmsc{}  is solvable in time $2^{\Oh(k^2\log k)}\cdot n^{\Oh(1)}$
 on regular $n$-element matroids, where  $k=\ell-w(T)$.
\end{theorem}

The remaining part of the section contains the proof of the theorem.  For technical reasons, in our algorithm we solve a special variants of \wmsc. 
In particular, in our algorithm,   the information about  circuits in $M$ will be derived from circuits of size $3$. We need the following technical definition. 
\begin{definition}[\textbf{Circuit constraints and extensions}]
Let $M$ be 
 a binary matroid $M$, given together with a set of terminals $T\subseteq E(M)$, and 
 a family  $\mathcal{X}$ of pairwise disjoint circuits of $M$ of size $3$,  which are also disjoint with $T$.
Then a \emph{\constrext}  for $M, T$ and $\mathcal{X}$ is  an 8-tuple   $(M,T,\mathcal{X},P,\mathcal{Z},w,\mathcal{W},k)$, 
where
\begin{itemize}
\item $P$ is a mapping   assigning  to each $X\in\mathcal{X}$ a nonempty set  $P(X)$  of subsets of $X$ of size 1 or 2, 
\item 
$\mathcal{Z}$ is either the empty set, or is a pair of the form $(Z,t)$, where $Z$ is a circuit of size 3 disjoint with the circuits of $\mathcal{X}$ and with terminals $T$, and $t$ is an element of $Z$,
\item $w$ is 
a weight function,  $w\colon E\setminus L\rightarrow \mathbb{N}$, where $L=\cup_{X\in\mathcal{X}}X$,   
\item 
$\mathcal{W}=\{w_X\mid X\in\mathcal{X}\}$ is a family of weight functions, where 
$w_X\colon P(X)\rightarrow\mathbb{N}$ for each $X\in \mathcal{X}$, and 
\item  $k$ is an integer.
\end{itemize}
We say that  circuit $C$ of $M$ is a  \emph{feasible extension satisfying \constrext  $(M,T,\mathcal{X},P,\mathcal{Z},w,\mathcal{W},k)$} (or just feasible when it is clear from the context) if
\begin{itemize}
\item 
  $C\cap X\in P(X)$ for each $X\in \mathcal{X}$, 
  \item if $\mathcal{Z}\neq\emptyset$, then 
 $C\bigtriangleup Z$ is a circuit of $M$ and $Z\cap C=\{t\}$, and 
 \item 
$w(C\setminus (T\cup L))+\sum_{X\in\mathcal{X}}w_X(C\cap X)\leq k$.
\end{itemize}
\end{definition}

We consider the following auxiliary problem. 

\defparproblema{\ewmsc}%
{A   \constrext $(M,T,\mathcal{X},P,\mathcal{Z},w,\mathcal{W},k)$.
  }%
{$k$}
{Decide whether there is an extension satisfying the \constrext.}

\smallskip
\noindent
Notice that \wmsc{} parameterized by $k=\ell-w(T)$ is the special case of \ewmsc{} for $\mathcal{X}=\emptyset$ and $\mathcal{Z}=\emptyset$. We call a circuit $C$ satisfying the requirements of the problem, i.e. which is an extension satisfying the corrsponding \constrext, by   a \emph{solution}. 
 We also refer to the value
 $\omega(C)=w(C\setminus (T\cup L))+\sum_{X\in\mathcal{X}}w_X(C\cap X)$ as to the \emph{weight} of $C$.

In Section~\ref{sec:wmsc-basic} we solve \ewmsc{} on matroids of basic types, and in Section~\ref{sec:wmsc:reg} we construct the algorithm for regular matroids.

\subsection{Solving \wmsc{} on basic matroids}\label{sec:wmsc-basic} 
First, we consider matroids obtained from $R_{10}$ by deleting elements and adding parallel elements. 

\begin{lemma}\label{lem:r10-wmsc}
\ewmsc{} can be solved in polynomial time on the class of matroids that can be obtained from $R_{10}$ by adding parallel elements and deleting some elements. 
\end{lemma}

\begin{proof}
Let $(M,T,\mathcal{X},P,\mathcal{Z},w,\mathcal{W},k)$ be an instance of \ewmsc{}, where 
$M$ is a matroid obtained from $R_{10}$ by adding parallel elements and deleting some elements. Since $M$ has no circuit of odd size, $\mathcal{X}=\emptyset$ and $\mathcal{Z}=\emptyset$. If $e_1,e_2\in E\setminus T$ are parallel, then any circuit $C$ contains at most one of the elements $e_1,e_2$, and if $e_1\in C$, then $C'=(C\setminus\{e_1\})\cup\{e_2\}$ is a circuit by Observation~\ref{obs:par}. It means that we can apply the following reduction rule:

\begin{reduction}
If there are parallel $e_1,e_2\in E\setminus T$ and $w(e_1)\leq w(e_2)$, then delete $e_2$.
\end{reduction}

\smallskip
\noindent
The matroid obtained from $M$ by the exhaustive application of the rule has at most 10 nonterminal elements. Hence, the problem can be solved by brute force: for each set $S$ of nonterminal elements we check whether $S\cup T$ is a circuit and find a circuit of minimum weight it it exists. 
\end{proof}

To construct an algorithm for \ewmsc{} for graphic matroids, we consider the following parameterized problem:

\defparproblema{\ctse{}}%
{A graph $G$, a weight function $w\colon E(G)\rightarrow \mathbb{N}$, a set of terminals $T\subseteq E(G)$, and a positive integer $k$.}%
{$k$}
{Does $G$ have a cycle $C$ with $T\subseteq E(C)$ such that $w(E(C))-w(T)\leq k$?}

\smallskip
\noindent 
We show that \ctse is \classFPT.
This problem can be  solved in time $2^k n^{\cO(1)}$ by making use of the  
randomized algorithm of  Bj{\"{o}}rklund et al.~\cite{BjorklundHT12}. 
As the running time of our algorithms for \wmsc{} is dominated by the running time of the algorithm for cographic matroids, we give here  a
\emph{deterministic} algorithm with a  worse constant in the base of the exponent. The algorithm is  based of the color coding technique of 
Alon et al.~\cite{AlonYZ95}.

\begin{lemma}\label{lem:graph-cycle}
\ctse{} is  solvable in time $2^{\Oh(k)}\cdot n^{\Oh(1)}$.
\end{lemma}

\begin{proof}
Let $(G,w,T,k)$ be an instance of \ctse{}.
First, we exhaustively apply the following reduction rules.

\begin{reduction}[\bf Stopping Rule]\label{rule:ctse-ind}
If $G[T]$ is not a disjoint union of paths or $G[T]$ has at least $k+1$ components, then return a no-answer and stop.
\end{reduction}

\begin{reduction}[\bf Dissolving Rule]\label{rule:ctse-dissolve}
If there is a vertex $v$ incident to two distinct edges $vx,vy\in T$, then do the following:
\begin{itemize}
\item delete each edge $e\in E(G)\setminus T$ incident to $v$;
\item delete $v$ and replace $vx,vy$ by an edge $xy$ and set $w(xy)=1$; set $T=(T\setminus\{vx,vy\})\cup\{xy\}$. 
\end{itemize} 
\end{reduction}

It is straightforward to see that the rules are safe. Assume that we do not stop when applying Rule~\ref{rule:ctse-ind}, and, to simplify notations, we use $(G,w,T,k)$ to denote the instance obtained after applying {\bf Dissolving Rule}. Let $T=\{x_1y_1,\ldots,x_ry_r\}$. Notice that the edges of $T$ are independent, i.e, have no common end-vertices, and  
$r\leq k$. If $r=1$, then we find a shortest $(x_1,y_2)$-path in $G-x_1y_1$ using the Dijkstra's algorithm~\cite{Dijkstra59}. 
If the weight of the path is at most $k$, we are done. Otherwise, we have a no-instance.

 We assume from now that $r\geq 2$. 
Let $U=\{x_1,\ldots,x_r\}\cup\{y_1,\ldots,y_r\}$ and denote $h=k-r$. Observe that any cycle $C$ such that $T\subseteq E(C)$ and $w(E(C)\setminus T)\leq k$ has at most $k$ vertices and, therefore, at most $h$ vertices in $V(G)\setminus U$.  We use the color coding technique~\cite{AlonYZ95} to find a cycle $C$ of minimum weight with at most $k$ vertices such that $T\subseteq E(C)$.
First, we describe a randomized true-biased Monte-Carlo algorithm and then explain how to derandomize it.

We color the vertices of $V(G)\setminus U$ by $h$ colors uniformly  at random. Denote by $c(v)$ the color of $v\in V(G)\setminus U$.
Our aim is to find a \emph{colorful} cycle $C$ in $G$ of minimum weight such that $T\subseteq E(C)$ and the vertices of $V(C)\setminus U$ have distinct colors. 

First, for each set of colors $X\subseteq\{1,\ldots,h\}$, for each pair $\{i,j\}$ of distinct $i,j\in\{1,\ldots,r\}$ and each $u\in\{x_i,y_i\}$ and $v\in\{x_j,y_j\}$, we find a $(u,v)$-path $P$ of minimum weight such that $V(P)\cap U=\{u,v\}$ and the internal vertices of $P$ are colored by distinct colors from $X$. It can be done in a standard way by dynamic programming across subsets (see~\cite{AlonYZ95,CyganFKLMPPS15}). For completeness, we sketch how to find the weight of such a path.  

Denote for $z\in V(G)\setminus\{x_i,y_i\}$, by $s(X,u,z)$ the minimum weight of a $(u,z)$-path $P$ in $G$ with all internal vertices in $V(G)\setminus U$ such that $V(P)\setminus U$ are colored by distinct colors from $X$; we assume that $s(X,u,z)=+\infty$ if such a path does not exist. We also assume slightly abusing notations that $s(X,u,u)=0$ for any $X\subseteq\{1,\ldots,h\}$. 
Clearly, 
$$s(\emptyset,u,z)=
\begin{cases}
w(uz)&\mbox{if }uz\in E(G)\text{ and }z\in U\setminus\{x_i,y_i\},\\
+\infty&\mbox{otherwise}.
\end{cases}
$$
If $X\neq \emptyset$, it is straightforward to verify that for $z\in V(G)\setminus U$, $s(X,u,z)=$
$$=
\begin{cases}
\min\{s(X\setminus\{c(z)\},u,x)+w(xz)\mid xz\in E(G),x\in (V(G)\setminus U)\cup\{u\}\} &\mbox{if }c(z)\in X,\\
+\infty&\mbox{if }c(z)\notin X,
\end{cases}
$$
and for $z\in U\setminus\{x_i,y_i\}$, 
$$s(X,u,z)=\min\{s(X,u,x)+w(xz)\mid xz\in E(G), x\in (V(G)\setminus U)\cup\{u\}\}.
$$
Using these recurrences, we compute $s(X,u,v)$ for all $X\in \{1,\ldots,h\}$ and $v\in U\setminus\{x_i,y_i\}$ in time $2^h\cdot n^{\Oh(1)}$. We do these computations for all $u\in \{x_i,y_i\}$ for every $i\in\{1,\ldots,r\}$.  

Using the table of values of $s(X,u,v)$, we compute the table of values of $c'(X,Y,v)$ for $v\in\{x_i,y_i\}$, where $i\in \{2,\ldots,r\}$, $X\subseteq\{1,\ldots,h\}$ and $Y\in \{2,\ldots,r\}\setminus\{i\}$, where 
$c'(X,Y,v)$ is a minimum weight of a $(y_1,v)$-path $P$ in $G$ such that $E(P)\cap T=\{x_jy_j\mid j\in X\}$ and the vertices $V(P)\setminus U$ are colored by distinct colors from $X$.
Notice that $c'(\{1,\ldots,h\},\{2,\ldots,r\},y_1)$ is the minimum weight of a cycle $C$ containing the edges of $T$ with $|V(C)\setminus U|\leq h$. 
For $Y=\emptyset$, 
$$c'(X,Y,v)=c(X,y_1,v).$$
For $Y\neq\emptyset$, we have that 
\begin{align*}
c'(X,Y,v)=\min\{\min\{&c'(X\setminus X',Y\setminus\{j\},x_j)+w(x_jy_j)+c(X',y_j,v),\\
&c'(X\setminus X',Y\setminus\{j\},y_j)+w(x_jy_j)+c(X',x_j,v)\}\mid X'\subseteq X,j\in\{1,\ldots,r\}\}.
\end{align*}
The correctness of the recurrence is proved by the standard arguments. We obtain that the table of values of $c'(X,Y,v)$ can be constructed in time $2^h2^r\cdot n^{\Oh(1)}$. Hence,  
$c'(\{1,\ldots,h\},\{2,\ldots,r\},y_1)$ can be computed in time $2^k\cdot n^{\Oh(1)}$.

We have that in time $2^k\cdot n^{\Oh(1)}$ we can check whether we have a \emph{colorful solution}, i.e., a cycle $C$ of weight at most $w(T)+k$ such that $T\subseteq E(C)$ and the vertices of 
$V(C)\setminus U$ are colored by distinct colors. If we have a colorful solution, then we return it. 

Notice that if $C$ is a solution for $(G,w,T,k)$, that is, $T\subseteq E(C)$ and $w(E(C)\setminus T)\leq k$, then the probability that the vertices of $V(C)\setminus U$ are colored by distinct colors from the set $\{1,\ldots,h\}$ is at least $h!/h^h\geq e^{-k}$. Hence, it is sufficient to repeat the algorithm for $e^k$ random colorings to claim that the probability that $(G,w,T,k)$ has a solution but our algorithm returns a no-answer for $e^k$ random colorings is at most $(1-1/e^k)^{e^k}\leq 1/e$, that is, we have a true biased  Monte-Carlo \classFPT{} algorithm that runs in time
$(2e)^k\cdot n^{\Oh(1)}$.

This algorithm can be derandomized by the standard tools~\cite{AlonYZ95,CyganFKLMPPS15} by replacing the random colorings by \emph{perfect hash functions}. The currently best family of 
perfect hash functions is constructed by Naor et al.  in~\cite{NaorSS95}. 
\end{proof}

\begin{lemma}\label{lem:graphic-wmsc}
\ewmsc{} can be solved in time $2^{\Oh(k)}\cdot |E(M)|^{\Oh(1)}$ on graphic matroids. 
\end{lemma}

\begin{proof}
Let $(M,T,\mathcal{X},P,\mathcal{Z},w,\mathcal{W},k)$ be an instance of \ewmsc{}, where $M$ is a graphic matroid.  We find $G$ such that $M$ is isomorphic to $M(G)$ in polynomial time using the results of Seymour~\cite{Seymour81} and assume that $M=M(G)$. Notice that we can assume without loss of generality that $G$ is connected. We reduce the problem to \ctse{}.
If $|\mathcal{X}|>k$, then we have a trivial no-instance. Assume from now that $|\mathcal{X}|\leq k$ and let $\mathcal{X}=\{X_1,\ldots,X_r\}$. 

First, we solve the problem for the case $\mathcal{Z}=\emptyset$. 
If $C$ is a solution, then $C\cap X\in P(X)$ for $X\in\mathcal{X}$. For each $X_i\in\mathcal{X}$, we guess a set $Y_i\in P(X_i)$ such that $C\cap X_i=Y_i$ for a hypothetic solution $C$. Since $Y_i$ has size 1 or 2, we have at most $6^k$ possibilities to guess $Y_1,\ldots,Y_r$. If $\sum_{i=1}^rw_{X_i}(Y_i)>k$, then we discard the guess. Assume that $\sum_{i=1}^rw_{X_i}(Y_i)\leq k$.
We define the graph $G'=G-\cup_{i=1}^r(X_i\setminus Y_i)$, $T'=T\cup(\cup_{i=1}^r Y_i)$ and $k'=k-\sum_{i=1}^rw_{X_i}(Y_i)$. We also define $w'(e)=w(e)$ for $e\in E(G')\setminus T'$ and set $w'(e)=1$ for $e\in T'$. Then we solve \ctse{} for $(G',w',T',k')$ using Lemma~\ref{lem:graph-cycle}. It is straightforward to see that we have a solution $C$ for the considered instance of \ewmsc{}  such that $C\cap X_i=Y_i$ for $i\in\{1,\ldots,r\}$ if and only if $(G',w',T',k')$ is a yes-instance of \ctse{}. 

Assume now that $\mathcal{Z}=(Z,t)$. Clearly, $Z$ induces a cycle in $G$. Let $u$ be a vertex of this cycle that is not incident to the edge $t$. We construct the instances of \ctse{} for every guess of $Y_1,\ldots,Y_r$ in almost the same way as before. The difference is that we delete $u$ from the obtained graph, define $t$ to be a terminal and reduce the parameter by $w(t)$. Notice that if a terminal is incident to $u$, we have a no-instance for the considered guess.
 Lemma~\ref{lem:gr-cogr} i) immediately implies the correctness of the reduction. 

Since \ctse{} can be solved in time $2^{\Oh(k)}\cdot n^{\Oh(1)}$ by Lemma~\ref{lem:graph-cycle} for each constructed instance and we consider at most $6^k$ instances and each instance is constructed in polynomial time, the total running time is $2^{\Oh(k)}\cdot n^{\Oh(1)}$. Because $G$ is connected, we can write the running time as $2^{\Oh(k)}\cdot |E(M)|^{\Oh(1)}$.
\end{proof}

We use Theorem~\ref{thm:graph-cut} to solve \ewmsc{} on cographic matroids.

\begin{lemma}\label{lem:cographic-wmsc}
\ewmsc{} can be solved in time $2^{\Oh(k^2)}\cdot |E(M)|^{\Oh(1)}$ on cographic matroids. 
\end{lemma}

\begin{proof}
Let $(M,T,\mathcal{X},P,\mathcal{Z},w,\mathcal{W},k)$ be an instance of \ewmsc{}, where $M$ is a cographic matroid.  We find $G$ such that $M$ is isomorphic to $M^*(G)$ in polynomial time using the results of Seymour~\cite{Seymour81} and assume that $M=M(G)$. Notice that we can assume without loss of generality that $G$ is connected. We reduce the problem to \emc{}.

If $|\mathcal{X}|>k$, then we have a trivial no-instance. Assume from now that $|\mathcal{X}|\leq k$ and let $\mathcal{X}=\{X_1,\ldots,X_r\}$. 
If $C$ is a solution, then $C\cap X\in P(X)$ for $X\in\mathcal{X}$. For each $X_i\in\mathcal{X}$, we guess a set $Y_i\in P(X_i)$ such that $C\cap X_i=Y_i$ for a hypothetic solution $C$. Since $Y_i$ has size 1 or 2, we have at most $6^k$ possibilities to guess $Y_1,\ldots,Y_r$. If $s=\sum_{i=1}^rw_{X_i}(Y_i)>k$, then we discard the guess. If $\mathcal{Z}=(Z,t)$ and 
$s+w(t)>k$, then we also can discard the guess. Assume that it is not the case.
We construct $G'$ by contracting the edges of $\cup_{i=1}^r(X_i\setminus Y_i)$; for simplicity, we do not distinguish the edges of $G'$ obtained by contractions from the edges of the original graph. If $\mathcal{Z}=\emptyset$, then we set $T'=T\cup(\cup_{i=1}^r Y_i)$, $R_1=\emptyset$ and $k'=k-s$, and if $\mathcal{Z}=(Z,t)$, then 
$T'=T\cup\cup_{i=1}^r Y_i\cup\{t\}$, $R_1$ is defined to be the set containing the end-vertices of the edges of $Z\setminus\{t\}$ and $k'=k-s-w(t)$. We also define $w'(e)=w(e)$ for $e\in E(G')\setminus T'$ and set $w'(e)=1$ for $e\in T'$. 
Then we solve \emc{} for $(G',w',T',R_1,\emptyset,k')$ using Theorem~\ref{thm:graph-cut}. If $\mathcal{Z}=\emptyset$,
then it is straightforward to see that we have a solution $C$ for the considered instance of \ewmsc{}  such that $C\cap X_i=Y_i$ for $i\in\{1,\ldots,r\}$ if and only if $(G',w',T',k')$ is a yes-instance of \ctse{}. If $\mathcal{Z}=(Z,t)$, then correctness follows from Lemma~\ref{lem:gr-cogr} ii).

Since \emc{} can be solved in time $2^{\Oh(k^2\log k)}\cdot n^{\Oh(1)}$ by Theorem~\ref{thm:graph-cut} for each constructed instance and we consider at most $6^k$ instances and each instance is constructed in polynomial time, the total running time is $2^{\Oh(k^2\log k)}\cdot n^{\Oh(1)}$. Because $G$ is connected, we can write the running time as $2^{\Oh(k^2\log k)}\cdot |E(M)|^{\Oh(1)}$.
\end{proof}

\subsection{Proof of Theorem~\ref{thm:wmsc}}\label{sec:wmsc:reg}
Now we are ready to give an algorithm for \wmsc{} parameterized by $k=\ell-w(T)$ on regular matroids. 
Let $(M,w,T,\ell)$ be an instance of \wmsc{}, where $M$ is regular. We consider it to be an instance  $(M,T,\mathcal{X},P,\mathcal{Z},w,\mathcal{W},k)$ of \ewmsc, where $\mathcal{X}=\emptyset$ and $\mathcal{Z}=\emptyset$.
If $M$ can be obtained from $R_{10}$ by the addition of parallel elements or 
is graphic or cographic, we solve the problem directly using Lemmas~\ref{lem:r10-wmsc}--\ref{lem:cographic-wmsc}. Assume that it is not the case. 
Using Theorem~\ref{thm:decomp-good},  we find a conflict tree $\mathcal{T}$. 
We select a node $r$ of $\mathcal{T}$ containing an element of $T$ as a root.

We say that an instance $(M,T,\mathcal{X},P,\mathcal{Z},w,\mathcal{W},k)$ of \ewmsc{} is \emph{consistent (with respect to $\mathcal{T}$)} if $\mathcal{Z}=\emptyset$ and for any $X\in\mathcal{X}$, $X\in E(M)$ for some $M\in\mathcal{M}$. Clearly, the instance obtained from the original input instance  $(M,w,T,\ell)$ of \wmsc{} is consistent. We use reduction rules that remove leaves keeping this property.

Let $M_\ell\in\mathcal{M}$ be a matroid that is a leaf of $\mathcal{T}$. We construct reduction rules depending on whether  $M_\ell$ is 1, 2 or 3-leaf.
Denote by $M_s$ its neighbor in $\mathcal{T}$. 
Let also $\mathcal{T}'$ be the tree obtained from $\mathcal{T}$ be the deletion of $M_\ell$, and let $M'$ be the matroid defined by $\mathcal{T}'$. Clearly, $M=M'\oplus M_\ell$.

Throughout this section, we say that a reduction rule is \emph{safe} if it either correctly solves the a problem or returns an equivalent instance of \ewmsc{} together with corresponding conflict tree of the obtained matroid that is consistent and the value of the parameter does not increase.  

From now, let $(M,T,\mathcal{X},P,\mathcal{Z},w,\mathcal{W},k)$ be a consistent instance of \ewmsc. Denote $L=\cup_{X\in\mathcal{X}}X$.

\begin{reduction}[{\bf $1$-Leaf reduction rule}]\label{rule:1-leaf-wmsc}
If $M_\ell$ is a 1-leaf, then do the following.
\begin{itemize}
\item[i)] If $T\cap E(M_\ell)\neq \emptyset$ or there is $X\in\mathcal{X}$ such that $X\subseteq E(M_\ell)$, then stop and return a no-answer,
\item[ii)] Otherwise, return the instance  $(M',T,\mathcal{X},P,\emptyset,w',\mathcal{W},k)$, where $w'$ is the restriction of $w$ on $E(M')\setminus L$, and solve it using the conflict tree $\mathcal{T}'$.
\end{itemize}
\end{reduction}

Since the root matroid contains at least one terminal, Lemma~\ref{lem:circ} i) immediately implies the following lemma.

\begin{lemma}\label{lem:1-leaf-wmsc}
Reduction Rule~\ref{rule:1-leaf-wmsc}  is safe and can be implemented to run  in time polynomial in $|E(M)|$. 
\end{lemma}

\begin{reduction}[{\bf $2$-Leaf reduction rule}]\label{rule:2-leaf-wmsc}
If $M_\ell$ is a 2-leaf, then let $\{e\}=E(M_\ell)\cap E(M_s)$ and do the following. 
\begin{itemize}
\item[i)] If $T\cap E(M_\ell)= \emptyset$ and there is no $X\in\mathcal{X}$ such that $X\subseteq E(M_\ell)$, then find a circuit $C_\ell$ of $M_\ell$ containing $e$ with minimum $w(C_\ell\setminus \{e\})\leq k$. If there is no such a circuit, then set $w'(e)=k+1$, and let $w'(e)=w(C_\ell\setminus \{e\}$ otherwise. Assume that $w'(e')=w(e')$ for $e'\in E(M')\setminus L$. 
Return the instance  $(M',T,\mathcal{X},P,\emptyset,w',\mathcal{W},k)$
 and solve it using the conflict tree $\mathcal{T}'$.
\item[ii)] Otherwise, if $T\cap E(M_\ell)\neq\emptyset$ or there is  $X\in\mathcal{X}$ such that $X\subseteq E(M_\ell)$, consider $T_\ell=(T\cap E(M_\ell))\cup\{e\}$ and
$\mathcal{X}_\ell=\{X\in\mathcal{X}\mid X\subseteq E(M_\ell)\}$. Define $P_\ell$, $w_\ell$, $\mathcal{W}_\ell$ by restricting the corresponding functions by $E(M_\ell)$ assuming additionally that $w_\ell(e)=1$.
Find the minimum $k_\ell\leq k$ such that $(M_\ell,T_\ell,\mathcal{X}_\ell,P_\ell,\emptyset,w_\ell,\mathcal{W}_\ell,k_\ell)$ is a yes-instance of \ewmsc{}. Stop and return a no-answer if such $k_\ell$ does not exist. 
 Otherwise, do the following.
Set $T'=(T\cap E(M'))\cup \{e\}$ and $\mathcal{X}'=\{X\in\mathcal{X}\mid X\subseteq E(M')\}$. 
 Define $P'$, $w'$, $\mathcal{W}'$ by restricting the corresponding functions by $E(M')$ assuming additionally that $w'(e)=1$.
Return the instance  $(M',T',\mathcal{X}',P',\emptyset,w',\mathcal{W}',k-k_\ell)$ and solve it using the conflict tree $\mathcal{T}'$.
\end{itemize}
\end{reduction}

\begin{lemma}\label{lem:2-leaf-wmsc}
Reduction Rule~\ref{rule:2-leaf-wmsc}  is safe and can be implemented to run  in time $2^{\Oh(k^2\log k)}\cdot|E(M)|^{\Oh(1)}$.
\end{lemma}

\begin{proof}
Clearly, if the rule returns a new instance, then it is consistent with respect to $\mathcal{T}'$ and the parameter does not increase.

We show that the rule either correctly solves the problem or returns an equivalent instance. 

Suppose that $(M,T,\mathcal{X},P,\mathcal{Z},w,\mathcal{W},k)$ is a consistent yes-instance. We prove that the rule returns a yes-instance. Denote by $C$ a circuit of $M$ that is a solution for the instance.
We consider two cases corresponding to the cases i) and ii) of the rule.

\medskip
\noindent
{\bf Case 1.}  $T\cap E(M_\ell)= \emptyset$ and there is no $X\in\mathcal{X}$ such that $X\in E(M_\ell)$. If $C\subseteq E(M')$, then by Lemma~\ref{lem:circ} ii), $C$ is a circuit of $M'$. 
It is straightforward to see that $C$ is a solution for $(M',T,\mathcal{X},P,\emptyset,w',\mathcal{W},k)$.
 Suppose that $C\cap E(M_\ell)\neq\emptyset$. Then $C=C_1\bigtriangleup C_2$, where $C_1\in \mathcal{C}(M')$, $C_2\in\mathcal{C}(M_2)$ and $e\in C_1\cap C_2$ by Lemma~\ref{lem:circ} ii). 
It remains to observe that $C_1$ is a feasible circuit 
 for  $(M',T,\mathcal{X},P,\emptyset,w',\mathcal{W},k)$ and its weight is at most the weight of $C$. Hence, $C_1$ is a solution for $(M',T,\mathcal{X},P,\emptyset,w',\mathcal{W},k)$ and the algorithm returns  is a yes-instance.

\medskip
\noindent
{\bf Case 2.}  $T\cap E(M_\ell)\neq\emptyset$ or there is  $X\in\mathcal{X}$ such that $X\subseteq E(M_\ell)$. Then by Lemma~\ref{lem:circ} ii), $C=C_1\bigtriangleup C_2$, where $C_1\in \mathcal{C}(M')$, $C_2\in\mathcal{C}(M_2)$ and $e\in C_1\cap C_2$. 
We have that $C_2$ is a solution for $(M_\ell,T_\ell,\mathcal{X}_\ell,P_\ell,\emptyset,w_\ell,\mathcal{W}_\ell,k')$, where $k'\leq k$ is the weight of $C_2$ and the algorithm does not stop. Also we have that $C_1$ is a solution for $(M',T',\mathcal{X}',P',\emptyset,w',\mathcal{W}',k-k_\ell)$ as $C_1$ is feasible and its weight is $\omega(C)-k'\leq k-k_\ell$, i.e., the rule returns a yes-instance.

\medskip
Suppose now that the instance constructed by the rule is a yes-instance with a solution $C'$. We show that the original instance $(M,T,\mathcal{X},P,\mathcal{Z},w,\mathcal{W},k)$ is a  yes-instance. We again consider two cases.

\medskip
\noindent
{\bf Case 1.} The new instance is constructed by Rule~\ref{rule:2-leaf-wmsc} i). 
If $e\notin C'$, then $C'$ is a circuit of $M$ by Lemma~\ref{lem:circ} ii) and, therefore, $C'$ is a solution for  $(M,T,\mathcal{X},P,\mathcal{Z},w,\mathcal{W},k)$, that is,  the original instance is a yes-instance. Assume that $e\in C'$. In this case, $w'(e)\leq k$. Hence, there is a circuit $C''$ of $M_\ell$ containing $e$ with $w(C''\setminus\{e\})=w'(e)$. By Lemma~\ref{lem:circ} ii), $C=C'\bigtriangleup C''$ is a circuit of $M$. We have that $C$ is a solution for $(M,T,\mathcal{X},P,\mathcal{Z},w,\mathcal{W},k)$ and it is a yes-instance.

\medskip
\noindent
{\bf Case 2.} The new instance is constructed by Rule~\ref{rule:2-leaf-wmsc} ii).
In this case, $(M_\ell,T_\ell,\mathcal{X}_\ell,P_\ell,\emptyset,w_\ell,\mathcal{W}_\ell,k_\ell)$ has a solution $C''$ of weight $k_\ell$. Notice that $e\in C'\cap C''$. We have that $C=C'\bigtriangleup C''$ is a circuit of $M$ by Lemma~\ref{lem:circ} ii). We have that $C$ is a solution for$(M,T,\mathcal{X},P,\mathcal{Z},w,\mathcal{W},k)$ and, therefore, the original instance is a yes-instance.

\medskip
We proved that the rule is safe. To evaluate the running time, notice first that we can find a a circuit $C_\ell$ of $M_\ell$ containing $e$ with minimum $w(C_\ell\setminus \{e\})\leq k$
in Rule~\ref{rule:2-leaf-wmsc} i) in  time $2^{\Oh(k^2\log k)}\cdot|E(M_\ell)|^{\Oh(1)}$
 using the observation that we have an instance of \wmsc{} with $T=\{e\}$ and can apply Lemmas~\ref{lem:r10-wmsc}--\ref{lem:cographic-wmsc} depending on the type of $M_\ell$\footnote{In fact, it can be done in polynomial time for this degenerate case}. 
We find $k_\ell$ in Rule~\ref{rule:2-leaf-wmsc} ii) by solving 
$(M_\ell,T_\ell,\mathcal{X}_\ell,P_\ell,\emptyset,w_\ell,\mathcal{W}_\ell,k_\ell)$ for $k_\ell\leq k$ using  Lemmas~\ref{lem:r10-wmsc}--\ref{lem:cographic-wmsc} depending on the type of $M_\ell$ in time 
$2^{\Oh(k^2\log k)}\cdot|E(M_\ell)|^{\Oh(1)}$. 
\end{proof}

\begin{reduction}[{\bf $3$-Leaf reduction rule}]\label{rule:3-leaf-wmsc}
If $M_\ell$ is a 3-leaf, then let
 $S=\{e_1,e_2,e_3\}=E(M_\ell)\cap E(M_s)$ and do the following. 
\begin{itemize}
\item[i)] If $T\cap E(M_\ell)= \emptyset$ and there is no $X\in\mathcal{X}$ such that $X\subseteq E(M_\ell)$, then for each $i\in\{1,2,3\}$, find a 
 circuit $C_\ell^{(i)}$ of $M_\ell$ such that  $C_\ell^{(i)}\cap S=\{e_i\}$ and $C_\ell^{(i)}\bigtriangleup S$ is a circuit of $M_\ell$
with minimum $w(C_\ell^{(i)}\setminus \{e_i\})\leq k$. If there is no such a circuit, then set $w'(e_i)=k+1$, and let $w'(e_i)=w(C_\ell^{(i)}\setminus \{e_i\})$ otherwise. Assume that $w'(e')=w(e')$ for $e'\in E(M')\setminus (L\cup S)$. 
Return the instance  $(M',T,\mathcal{X},P,\emptyset,w',\mathcal{W},k)$
 and solve it using the conflict tree $\mathcal{T}'$.
\item[ii)] If  there is no $X\in\mathcal{X}$ such that $X\subseteq E(M_\ell)$, but $T_\ell=T\cap E(M_\ell)\neq\emptyset$
and there is $i\in\{1,2,3\}$ such that $C_\ell=T_\ell\cup\{e_i\}$ is a circuit of $M_\ell$, then consider two cases.
\begin{itemize}
\item $C_\ell\bigtriangleup S$ is a circuit of $M_\ell$. Set $w'(e_i)=1$ and  assume that $w'(e')=w(e')$ for $e'\in E(M')\setminus (S\cup L)$.
For each $j\in\{1,2,3\}\setminus\{i\}$, do the following.  Let $h\in\{1,2,3\}\setminus\{i,j\}$.
Set $\mathcal{X}_\ell=\{S\}$, $P_\ell(S)=\{e_j\}$, $w_S(\{e_h\})=1$ and $\mathcal{W}_\ell=\{w_S\}$. Let $w_\ell$ be a restriction of $w$ on $E(M_\ell)$.
Find a minimum $k_\ell^{(h)}\leq k+1$ such that $(M_\ell,T_\ell,\mathcal{X}_\ell,P_\ell,\emptyset,w_\ell,\mathcal{W}_\ell,k_\ell^{(h)})$ is a yes-instance of \ewmsc{}.
If there is no such $k_\ell^{(h)}$, then set $w'(e_j)=k+1$ and set $w'(e_j)=k_\ell^{(h)}-1$ otherwise.
Set $T'=(T\cap E(M'))\cup\{e_i\}$. 
Return the instance  $(M',T',\mathcal{X},P,\emptyset,w',\mathcal{W},k)$
 and solve it using the conflict tree $\mathcal{T}'$.
\item $C_\ell\bigtriangleup S$ is not a circuit of $M_\ell$. Set $w'(e_i)=k+1$ and $w'(e_j)=1$ for $j\in\{1,2,3\}\setminus\{i\}$.
 Assume that $w'(e')=w(e')$ for $e'\in E(M')\setminus (L\cup S)$.
Set $T'=(T\cap E(M'))\cup(S\setminus\{e_i\})$. 
Return the instance  $(M',T',\mathcal{X},P,\emptyset,w',\mathcal{W},k)$
 and solve it using the conflict tree $\mathcal{T}'$.
\end{itemize}
\item[iii)] Otherwise, let $T_\ell=T\cap E(M_\ell)$ and $\mathcal{X}_\ell=\{X\in\mathcal{X}\mid X\subseteq E(M_\ell)\}$. Define $P_\ell$, $w_\ell$, $\mathcal{W}_\ell$ by restricting the corresponding functions by $E(M_\ell)$. Construct the set $Y$ of subsets of $S$ and the function $w_S\colon Y\rightarrow \mathbb{N}$ as follows. Initially, set $Y=\emptyset$.
\begin{itemize}
\item  Define $w_\ell'(e_i)=1$ for $i\in\{1,2,3\}$ and let $w_\ell'(e)=w_\ell(e)$ for $e\in E(M_\ell)\setminus (L\cup S)$.
 For $i\in\{1,2,3\}$, 
find the minimum $k_\ell^{(i)}\leq k+1$ such that $(M_\ell,T_\ell,\mathcal{X}_\ell,P_\ell,(S,e_i),w_\ell',\mathcal{W}_\ell,k_\ell^{(i)})$ is a yes-instance of \ewmsc{}.
If such $k_\ell^{(i)}$ exists, then add $\{e_i\}$ in $Y$ and set $w_S(\{e_i\})=k_\ell^{(i)}-1$.
\item Let $\mathcal{X}_\ell'=\mathcal{X}_\ell\cup \{S\}$. 
For each $i\in\{1,2,3\}$, do the following. Set $P_\ell^{(i)}(X)=P_\ell(X)$ for $X\in \mathcal{X}_\ell$ and $P_\ell^{(i)}(Y)=\{x_i\}$, set $w_S^{(i)}(\{e_i\})=1$ and $\mathcal{W}_\ell^{(i)}=\mathcal{W}_\ell\cup\{w_S^{(i)}\}$. 
Find the minimum $k_\ell^{(i)}\leq k+1$ such that $(M_\ell,T_\ell,\mathcal{X}_\ell',P_\ell^{(i)},\emptyset,w_\ell,\mathcal{W}_\ell^{(i)},k_\ell^{(i)})$ is a yes-instance of \ewmsc{}.
If such $k_\ell^{(i)}$ exists, then add $S\setminus \{e_i\}$ in $Y$ and set $w_S(S\setminus\{e_i\})=k_\ell^{(i)}-1$. 
\end{itemize}
If $Y=\emptyset$, then return a no-answer and stop. Otherwise, set $T'=T\cap E(M')$, $\mathcal{X}'=\{X\in\mathcal{X}\mid X\subseteq E(M')\}\cup \{S\}$ and for $X\in\mathcal{X}'$, let $P'(X)=P(X)$ if $X\subseteq P(X)$ and $P'(S)=Y$. Also let $\mathcal{W}'=\{w_X\mid X\in \mathcal{X}'\}$ and let $w'$ be the restriction of $w$ on $E(M')$.
Return the instance  $(M',T',\mathcal{X}',P',\emptyset,w',\mathcal{W}',k)$ and solve it using the conflict tree $\mathcal{T}'$.
\end{itemize}
\end{reduction}

\begin{lemma}\label{lem:3-leaf-wmsc}
Reduction Rule~\ref{rule:3-leaf-wmsc}  is safe and can be implemented to run  in time $2^{\Oh(k^2\log k)}\cdot|E(M)|^{\Oh(1)}$.
\end{lemma}

\begin{proof}
It is straightforward to see that if the rule returns a new instance, then it is consistent with respect to $\mathcal{T}'$ and the parameter does not increase.
We show that the rule either correctly solves the problem or returns an equivalent instance. 

Suppose that $(M,T,\mathcal{X},P,\mathcal{Z},w,\mathcal{W},k)$ is a consistent yes-instance. We prove that the rule returns a yes-instance. Denote by $C$ a circuit of $M$ that is a solution for the instance.
We consider three cases corresponding to the cases i)--iii) of the rule.

\medskip
\noindent
{\bf Case 1.} $T\cap E(M_\ell)= \emptyset$ and there is no $X\in\mathcal{X}$ such that $X\in E(M_\ell)$. 

 If $C\subseteq E(M')$, then by Lemma~\ref{lem:circ} iii), $C$ is a circuit of $M'$, and $C$ is a solution for the instance $(M',T,\mathcal{X},P,\emptyset,w',\mathcal{W},k)$ returned by Rule~\ref{rule:3-leaf-wmsc} i), that is, we get a yes-instance.
Suppose that $C\cap E(M_\ell)\neq\emptyset$. Then, by Lemma~\ref{lem:circ} iii),
$C=C_1\bigtriangleup C_2$, where $C_1\in \mathcal{C}(M')$, $C_2\in\mathcal{C}(M_\ell)$, $C_1\cap S=C_2\cap S=\{e_i\}$ for some $i\in\{1,2,3\}$, and 
$C_1\bigtriangleup S$ is a circuit of $M'$ or $C_2\bigtriangleup S$ is a circuit of $M_\ell$.

Suppose that $C_2\bigtriangleup S$ is a circuit of $M_\ell$. Then $C_2$ is a circuit of $M_\ell$ containing $e_i$ such that $C_2\bigtriangleup S$ is a circuit and $w(C_2\setminus \{e_i\})\leq k$.
We have that  $w'(e_i)\leq w(C_2\setminus \{e_i\})$. Hence, $C_1$  is a solution for the instance $(M',T,\mathcal{X},P,\emptyset,w',\mathcal{W},k)$ returned by Rule~\ref{rule:3-leaf-wmsc} i) and, therefore, the rule returns a yes-instance.

Assume now that $C_2\bigtriangleup Z$ is a not circuit of $M_\ell$. By Lemma~\ref{lem:circ-triangle}, $C_2$ is a disjoint union of two circuits $C_2^{(1)}$ and $C_2^{(2)}$ of $M_2$ containing $e_h,e_j\in Z\setminus\{e_i\}$, and  $C_2^{(1)}\bigtriangleup S$ and $C_2^{(2)}\bigtriangleup S$ are circuits of $M_\ell$.
We obtain that $w'(e_h)\leq  w(C_2^{(1)}\setminus \{e_h\})$ and  $w'(e_j)\leq  w(C_2^{(2)}\setminus \{e_j\})$. Consider $C_1'=C_1\bigtriangleup S$. Because  $C_2\bigtriangleup S$ is  not a circuit of $M_\ell$, $C_1'$ is a circuit of $M'$. Since $e_h,e_j\in E(M')$, we have that $C_1'$ is a solution for $(M',T,\mathcal{X},P,\emptyset,w',\mathcal{W},k)$ returned by Rule~\ref{rule:3-leaf-wmsc} i). Hence, we get a yes-instance.

\medskip
\noindent
{\bf Case 2.} There is no $X\in\mathcal{X}$ such that $X\subseteq E(M_\ell)$, but $T_\ell=T\cap E(M_\ell)\neq\emptyset$
and there is $i\in\{1,2,3\}$ such that $C_\ell=T_\ell\cup\{e_i\}$ is a circuit of $M_\ell$. 

Notice that $w'(e)\geq 1$ for $e\in E(M')\setminus L$, that is, the instance returned by \ref{rule:3-leaf-wmsc} ii) is a feasible instance of \ewmsc. To prove it, observe that if 
$C_\ell\bigtriangleup S$ is a circuit of $M_\ell$ and $j\in\{1,2,3\}\setminus \{i\}$, then $k_\ell^{(h)}\geq 2$, because any solution $C'$ for $(M_\ell,T_\ell,\mathcal{X}_\ell,P_\ell,\emptyset,w_\ell,\mathcal{W}_\ell,k_\ell^{(h)})$ contains at least one element of $E(M_\ell)\setminus (T_\ell\cup S)$. Otherwise, we get that $C_\ell\bigtriangleup C'=\{e_i,e_h\}$ is a cycle of $M_\ell$ contradicting that $S$ is a circuit of $M_\ell$.

By Lemma~\ref{lem:circ} iii),
$C=C_1\bigtriangleup C_2$, where $C_1\in \mathcal{C}(M')$, $C_2\in\mathcal{C}(M_\ell)$, $C_1\cap S=C_2\cap S=\{e_h\}$ for some $h\in\{1,2,3\}$, and 
$C_1\bigtriangleup S$ is a circuit of $M'$ or $C_2\bigtriangleup S$ is a circuit of $M_\ell$. 

Assume first that $C_\ell\bigtriangleup S$ is a circuit of $M_\ell$. If $h=i$, then it is straightforward to verify that $C'=C_1\bigtriangleup C_\ell$ is a solution for 
the instance 
$(M',T',\mathcal{X},P,\emptyset,w',\mathcal{W},k)$ returned by Rule~\ref{rule:3-leaf-wmsc} ii) and, therefore, the rule returns a yes-instance.
Suppose that $h\in\{1,2,3\}\setminus\{i\}$. 
We have that $C_2$ is a solution for 
$(M_\ell,T_\ell,\mathcal{X}_\ell,P_\ell,\emptyset,w_\ell,\mathcal{W}_\ell,k_\ell^{(h)})$ constructed in Rule~\ref{rule:3-leaf-wmsc} ii). 
Hence, $w'(e_j)=k_\ell^{(h)}-1$, where $k_\ell^{(h)}$ is at most the weight of the solution $C_2$ for  $(M_\ell,T_\ell,\mathcal{X}_\ell,P_\ell,\emptyset,w_\ell,\mathcal{W}_\ell,k_\ell^{(h)})$ and $j\in\{1,2,3\}\setminus \{i,h\}$.  
Notice that $C_\ell\subset C_2\bigtriangleup S$, that is, $C_2\bigtriangleup S$ is not a circuit of $M_\ell$. Hence, $C_1'=C_1\bigtriangleup S$ is a circuit of $M'$. 
We obtain that  
$C_1'$ is a solution for the instance $(M',T',\mathcal{X},P,\emptyset,w',\mathcal{W},k)$ returned by Rule~\ref{rule:3-leaf-wmsc} ii). Hence, we get  a yes-instance of the problem.

Suppose now that $C_\ell\bigtriangleup S$ is not a circuit of $M_\ell$. We claim that $h=i$ and $C_2=C_\ell$ in this case.
If $h=i$, then $C_2=C_\ell$ by minimality, because $T_\ell\subseteq C_2$. Suppose that $h\neq i$. By Lemma~\ref{lem:circ-triangle}, $C_\ell\bigtriangleup S$ is disjoint union of two circuits $C_\ell^{(1)}$ and $C_\ell^{(2)}$
of $M_\ell$ containing $e_h$ and $e_j$ respectively, where $j\in\{1,2,3\}\setminus\{i,h\}$. Therefore, $C_\ell^{(1)}\subseteq C_2$ and, by minimality, $C_2=C_\ell^{(1)}$, but at least one terminal of $T_\ell$ is in $C_\ell^{(2)}$ contradicting $T_\ell\subseteq C_2$. Hence, $h=i$ and $C_2=C_\ell$. Then $C_1'=C_1\bigtriangleup S$ is a circuit of $M'$ and is a solution for 
the instance 
$(M',T',\mathcal{X},P,\emptyset,w',\mathcal{W},k)$ returned by Rule~\ref{rule:3-leaf-wmsc} ii) and, therefore, the rule returns a yes-instance.

\medskip
\noindent
{\bf Case 3.} Cases~1 and 2 do not apply, that is, we are in the conditions of Rule~\ref{rule:3-leaf-wmsc} iii).  By Lemma~\ref{lem:circ} iii),
$C=C_1\bigtriangleup C_2$, where $C_1\in \mathcal{C}(M')$, $C_2\in\mathcal{C}(M_\ell)$, $C_1\cap S=C_2\cap S=\{e_i\}$ for some $i\in\{1,2,3\}$, and 
$C_1\bigtriangleup S$ is a circuit of $M'$ or $C_2\bigtriangleup S$ is a circuit of $M_\ell$. Notice that if Rule~\ref{rule:3-leaf-wmsc} iii) returns an instance, then 
$w_S$ has only positive values, because it always holds that $k_\ell^{(i)}\geq 2$, since the conditions of Rule~\ref{rule:3-leaf-wmsc} ii) are not fulfilled.

Assume first that $C_2\bigtriangleup S$ is a circuit of $M_\ell$. Notice that $C_2$ is a feasible circuit for  $(M_\ell,T_\ell,\mathcal{X}_\ell,P_\ell,(S,e_i),w_\ell',\mathcal{W}_\ell,k_\ell)$ for $k_\ell\leq k$ and its weight with respect to this instance is at most $k$. Hence, $\{e_i\}\in Y\neq\emptyset$. It means that we do not stop while executing Rule~\ref{rule:3-leaf-wmsc} iii) and $w_S(\{e_i\})$ is at most the weight of $C_2$. It implies that $C_1$ is a solution for $(M',T',\mathcal{X}',P',\emptyset,w',\mathcal{W}',k)$ returned by Rule~\ref{rule:3-leaf-wmsc} iii), i.e., we obtain a yes-instance.

Suppose that $C_2\bigtriangleup S$ is not a circuit of $M_\ell$. Then $C_1\bigtriangleup S$ is a circuit of $M'$. We have that $C_2$ is a feasible circuit for   $(M_\ell,T_\ell,\mathcal{X}_\ell',P_\ell^{(i)},\emptyset,w_\ell,\mathcal{W}_\ell^{(i)},k_\ell^{(i)})$  for $k_\ell\leq k$ and its weight with respect to this instance is at most $k$. Hence, $S\setminus \{e_i\}\in Y\neq\emptyset$. Therefore, we do not stop and $w_S(S\setminus \{e_i\})$ is at most the weight of $C_2$. It implies that $C_1'=C_1\bigtriangleup S$ is a solution for  $(M_\ell,T_\ell,\mathcal{X}_\ell',P_\ell^{(i)},\emptyset,w_\ell,\mathcal{W}_\ell^{(i)},k_\ell^{(i)})$ returned by Rule~\ref{rule:3-leaf-wmsc} iii), that is, the rule returns a yes-instance.

\medskip
\medskip
Suppose now that the instance constructed by the rule is a yes-instance with a solution $C'$. We show that the original instance $(M,T,\mathcal{X},P,\mathcal{Z},w,\mathcal{W},k)$ is a  yes-instance.  

We consider three cases corresponding to the cases of the rule.

\medskip
\noindent
{\bf Case 1.} The new instance is constructed by Rule~\ref{rule:3-leaf-wmsc} i). If $C'\cap S=\emptyset$, then it is straightforward to see that $C'$ is a solution for the original instance and, therefore, $(M,T,\mathcal{X},P,\mathcal{Z},w,\mathcal{W},k)$ is a  yes-instance. Suppose that $C'\cap S\neq\emptyset$. Clearly, $|C'\cap S|\leq 2$.

Assume that $C'\cap S=\{e_i\}$ for some $i\in\{1,2,3\}$. Clearly, $w'(e_i)\leq k$. Hence, $M_\ell$ has a circuit $C''$ with $C''\cap S=\{e_i\}$ such that $w'(e_i)=w(C''\setminus\{e_i\})$ and $C''\bigtriangleup S$ is a circuit of $M_\ell$. By Lemma~\ref{lem:circ} iii), $C=C'\bigtriangleup C''$ is a circuit of $M$. We obtain that $C$ is a solution for $(M,T,\mathcal{X},P,\mathcal{Z},w,\mathcal{W},k)$, that is, it is a  yes-instance.

Suppose that $C'\cap S=\{e_i,e_j\}$ for distinct $i,j\in\{1,2,3\}$. Let $h\in\{1,2,3\}\setminus\{i,j\}$. We have that $w'(e_i)\leq k$ and $w'(e_j)\leq k$. It means, that $M_\ell$ has circuits $C_1''$ and $C_2''$ such that $C_1''\cap S=\{e_i\}$, $C_2''\cap S=\{e_j\}$  and $w'(e_i)=w(C_1''\setminus\{e_i\})$,  $w'(e_j)=w(C_2''\setminus\{e_i\})$. Consider $C''=C_1''\bigtriangleup C_2''$. By Observation~\ref{obs:symm}, $C''$ is a cycle of $M_\ell$. Then there is a circuit $C'''\subseteq C''$ of $M_\ell$ such that $C'''\cap S=\{e_h\}$. Notice that $w(C'''\setminus\{e_h\}\leq w'(e_i)+w'(e_j)$. By Lemma~\ref{lem:circ-triangle-a}, $C'\bigtriangleup S$ is a circuit of $M$. Let $C=(C'\bigtriangleup S)\bigtriangleup C'''$. By Lemma~\ref{lem:circ} iii), $C$ is a circuit of $M$.  We have that $C$ is a solution for $(M,T,\mathcal{X},P,\mathcal{Z},w,\mathcal{W},k)$ and, therefore, it is a  yes-instance.

\medskip
\noindent
{\bf Case 2.} The new instance is constructed by Rule~\ref{rule:3-leaf-wmsc} ii).  Recall that $C_\ell=T_\ell\cup\{e_i\}$ is a circuit of $M_\ell$.
Clearly, $1\leq |C'\cap S|\leq 2$.

Suppose first that $C_\ell\bigtriangleup S$ is a circuit of $M_\ell$. If $|C'\cap S|=1$, then $C'\cap S=\{e_i\}$. Then we obtain that $C=C'\bigtriangleup C_\ell$ is  a solution for $(M,T,\mathcal{X},P,\mathcal{Z},w,\mathcal{W},k)$ and it is a  yes-instance. Assume that $C'\cap S=\{e_i,e_j\}$ for $j\in\{1,2,3\}\setminus\{i\}$. Then $w'(e_j)\leq k$. 
Then there is a circuit  $C''$ of $M_\ell$ such that $C''\cap S=\{e_h\}$ for $h\in\{1,2,3\}\setminus\{i,j\}$ that is a solution of weight $w'(e_j)+1$ for $(M_\ell,T_\ell,\mathcal{X}_\ell,P_\ell,\emptyset,w_\ell,\mathcal{W}_\ell,k_\ell^{(h)})$ considered by Rule~\ref{rule:3-leaf-wmsc} ii). Notice that $C'\bigtriangleup S$ is a circuit of $M'$ by Lemma~\ref{lem:circ-triangle-a}. By Lemma~\ref{lem:circ} iii), we obtain that $C=(C'\bigtriangleup S)\bigtriangleup C''$ is a solution for $(M,T,\mathcal{X},P,\mathcal{Z},w,\mathcal{W},k)$ and, therefore, it is a  yes-instance.

Assume now that $C_\ell\bigtriangleup S$ is a not circuit of $M_\ell$. Then $C'\cap S=\{e_h,e_j\}$ for $\{h,j\}=\{1,2,3\}\setminus\{i\}$. By Lemma~\ref{lem:circ-triangle-a}, $C'\bigtriangleup S$ is a circuit of $M'$, and by Lemma~\ref{lem:circ} iii), we obtain that $C=(C'\bigtriangleup S)\bigtriangleup C_\ell$ is a solution for $(M,T,\mathcal{X},P,\mathcal{Z},w,\mathcal{W},k)$, that is, it is a  yes-instance.

\medskip
\noindent
{\bf Case 3.} The new instance is constructed by Rule~\ref{rule:3-leaf-wmsc} iii). We have that $C'\cap S\in Y$ for the set $Y$ constructed by the rule.

Assume that $C'\cap S=\{e_i\}$ for $i\in\{1,2,3\}$. Then $w_S(\{e_i\})\leq k$ and, therefore, there is a solution $C''$ of weight $k_\ell^{(i)}=w_S(\{e_i\})+1$ for the instance
$(M_\ell,T_\ell,\mathcal{X}_\ell,P_\ell,(S,e_i),w_\ell',\mathcal{W}_\ell,k_\ell^{(i)})$ constructed by the rule. Notice that $C''\bigtriangleup S$ is a circuit of $M_\ell$.
We obtain that $C=C'\bigtriangleup C''$ is a solution for $(M,T,\mathcal{X},P,\mathcal{Z},w,\mathcal{W},k)$ and it is a  yes-instance.

Suppose now that $C'\cap S=\{e_i,e_j\}$ for distinct $i,j\in\{1,2,3\}$. Let $h\in\{1,2,3\}\setminus\{i,j\}$. We have that $w_S(\{e_i,e_j\})\leq k$. Hence, there is a solution $C''$ of weight $k_\ell^{(h)}=w(\{e_i,e_j\})+1$ for the instance $(M_\ell,T_\ell,\mathcal{X}_\ell',P_\ell^{(i)},\emptyset,w_\ell,\mathcal{W}_\ell^{(i)},k_\ell^{(h)})$. By Lemma~\ref{lem:circ-triangle-a}, $C'\bigtriangleup S$ is a circuit of $M'$, and by Lemma~\ref{lem:circ} iii), $C=C'\bigtriangleup C''$ is a circuit of $M$. We have that $C$ is a solution for the original instance $(M,T,\mathcal{X},P,\mathcal{Z},w,\mathcal{W},k)$. Hence, it is a yes-instance.

\medskip
To complete the proof, it remains to evaluate the running time. Rule~\ref{rule:3-leaf-wmsc} i) can be executed in  time $2^{\Oh(k^2\log k)}\cdot|E(M)|^{\Oh(1)}$.\footnote{In fact, it can be done in polynomial time for this degenerate case.} To see it, observe that to compute $w'(e_i)$ for $i\in\{1,2,3\}$, we can solve \ewmsc{} for $(M_\ell,\emptyset,\emptyset,\emptyset,(S,e_i),w_\ell,k_\ell^{(i)})$ for $k_\ell^{(i)}\leq k$, where $w_\ell(e)=w(e)$ for $e\in E(M_\ell)\setminus (L\cup S)$ and $w_\ell(e_i)=1$ for $i\in\{1,2,3\}$, using Lemmas~\ref{lem:r10-wmsc}--\ref{lem:cographic-wmsc} depending on the type of $M_\ell$. For  Rule~\ref{rule:3-leaf-wmsc} ii), observe that it can be checked in polynomial time whether $C_\ell=T_\ell\cup\{e_i\}$ and $C_\ell\bigtriangleup S$ are circuits of $M$ for $i\in\{1,2,3\}$.  Then we can solve the problem for each $(M_\ell,T_\ell,\mathcal{X}_\ell,P_\ell,\emptyset,w_\ell,\mathcal{W}_\ell,k_\ell^{(h)})$ in time $2^{\Oh(k^2\log k)}\cdot|E(M)|^{\Oh(1)}$ by Lemmas~\ref{lem:r10-wmsc}--\ref{lem:cographic-wmsc}. Finally, the problem for every auxiliary instance $(M_\ell,T_\ell,\mathcal{X}_\ell,P_\ell,(S,e_i),w_\ell',\mathcal{W}_\ell,k_\ell^{(i)})$ and every
$(M_\ell,T_\ell,\mathcal{X}_\ell',P_\ell^{(i)},\emptyset,w_\ell,\mathcal{W}_\ell^{(i)},k_\ell^{(i)})$ can be solved in time $2^{\Oh(k^2\log k)}\cdot|E(M)|^{\Oh(1)}$ by Lemmas~\ref{lem:r10-wmsc}--\ref{lem:cographic-wmsc}.
\end{proof}

Now we can complete the proof of Theorem~\ref{thm:wmsc}.  Observe that $\mathcal{M}$ and the corresponding conflict tree $\mathcal{T}$ can be constructed in polynomial time by Theorem~\ref{thm:decomp-good}, and then we apply the reduction rules at most $|V(\mathcal{T})|-1$ times until we obtain an instance of \ewmsc{} for a matroid of one of basic types and solve the problem using Lemmas~\ref{lem:r10-wmsc}--\ref{lem:cographic-wmsc}.

\section{Solving \scir{} on regular matroids}\label{sec:scir}
In this section we prove the following theorem.

\begin{theorem}\label{thm:scir}
\scir{} is \classFPT{} on regular matroids when parameterized by $|T|$.
\end{theorem}

The remaining part of the section contains the proof of the theorem. 
Similarly to the proof of Theorem~\ref{thm:wmsc}, we solve a special variant of \scir{}. 
We redefine a simplified variant of \emph{circuit constraint} that we need in this section as follows.
\begin{definition}[\textbf{Circuit constraints and extensions}]
Let $M$ be 
 a binary matroid $M$ given together with 
a set $\mathcal{X}$ of nonempty pairwise disjoint subsets of $E(M)$ of size at most 3.
Then a \emph{\constrext}  for $M$ and $\mathcal{X}$ is  an 4-tuple   $(M,\mathcal{X},P,\mathcal{Z})$, 
where
\begin{itemize}
\item $P$ is a mapping   assigning  to each $X\in\mathcal{X}$ a nonempty set  $P(X)$  of subsets of $X$ of size 1 or 2, 
\item 
$\mathcal{Z}$ is either the empty set, or is a pair of the form $(Z,t)$, where $Z$ is a circuit of size 3 disjoint with the sets of $\mathcal{X}$ and  $t$ is an element of $Z$.
\end{itemize}
We say that  circuit $C$ of $M$ is a  \emph{feasible extension satisfying \constrext  $(M,\mathcal{X},P,\mathcal{Z})$} (or just feasible when it is clear from the context) if
\begin{itemize}
\item 
  $C\cap X\in P(X)$ for each $X\in \mathcal{X}$, and
  \item If $\mathcal{Z}\neq\emptyset$, then 
 $C\bigtriangleup Z$ is a circuit of $M$ and $Z\cap C=\{t\}$.
\end{itemize}
\end{definition}

\defproblema{\escir}%
{A   \constrext $(M,\mathcal{X},P,\mathcal{Z})$.
  }%
{Decide whether there is an extension satisfying the \constrext.}

\noindent\smallskip
We also say that a circuit $C$ is a  \emph{feasible extension satisfying \constrext  $(M,\mathcal{X},P,\mathcal{Z})$} is 
a \emph{solution} for an instance of \escir.
Clearly, \scir{} is a special case of \escir{} for $\mathcal{X}=\{\{t\}\mid t\in T\}$, $P(\{t\})=\{t\}$ for $t\in T$,
 and $\mathcal{Z}=\emptyset$.
In Section~\ref{sec:scir-basic} we construct algorithms for  \escir{} for basic matroids and in Section~\ref{sec:scir-reg} we explain how to use these results to solve \scir{} on regular matroids.

\subsection{Solving \escir{} on basic matroids}\label{sec:scir-basic} 
First, we consider matroids obtained from $R_{10}$ by deleting elements and adding parallel elements. Notice that, in fact, such matroids that occur in decompositions have bounded size but, formally, we have to deal with the possibility that the number of parallel elements added to $R_{10}$ can be arbitrary.

\begin{lemma}\label{lem:r10-scir}
\escir{} can be solved in polynomial time on the class of matroids that can be obtained from $R_{10}$ by adding parallel elements and deleting some elements. 
\end{lemma}

\begin{proof}
Let $(M,\mathcal{X},P,\mathcal{Z})$ be an instance of \escir, where $M$ is a matroid with a ground set $E$ that is obtained from $R_{10}$ be adding parallel elements and deleting some elements.
Notice that $\mathcal{Z}=\emptyset$, because $M$ has no circuits of odd size.

Notice that if $e$ and $e'$ are parallel elements of $M$, then for any circuit $C$ of $M$, either $C=\{e,e'\}$ or $|C\cap\{e,e'\}|\leq 1$.  
It implies that if $|\mathcal{X}|>10$, then $(M,\mathcal{X},P,\mathcal{Z})$ is a no-instance, because for any selection of sets $S(X)\in P(X)$, $\cup_{X\in\mathcal{X}}S(X)$ contains two parallel elements.
Suppose that this does not occur. Let $Y=\cup_{X\in\mathcal{X}}X$. 
 Let $M'$ be the matroid obtained from $M$ by the exhaustive deletions of  elements of $E\setminus Y$ that are parallel to some other remaining element of $E\setminus Y$. 
We claim that $(M,\mathcal{X},P,\mathcal{Z})$ is a yes-instance if and only if $(M',\mathcal{X},P,\mathcal{Z})$ is a yes-instance.
If $C$ is a circuit of $M'$ such that $T\subseteq C$, then $C$ is a circuit of $M$ as well. Hence, if $(M',\mathcal{X},P,\mathcal{Z})$ is a yes-instance, then $(M,\mathcal{X},\mathcal{P},Z)$ is a yes-instance of \escir. Suppose that  $(M,\mathcal{X},\mathcal{P},Z)$ is a yes-instance and let a circuit $C$ of $M$ be a solution for the instance  such that $|C\setminus E(M')|$ is minimum. 
If $C\subseteq E(M')$, then $C$ is a circuit of $M'$ and $(M',\mathcal{X},P,\mathcal{Z})$ is a yes-instance. Assume that there is $e\in C\setminus E(M')$. Then there is $e'\in E(M')$ that is parallel to $e$ in $M$ such that $e'\notin Y$.
Consider $C'=C\bigtriangleup\{e,e'\}$. By Observation~\ref{obs:par}, $C'$ is a circuit of $M$. 
We obtain that $C'$ is a solution such that
$|C'\setminus E(M')|<|C\setminus E(M')|$; a contradiction. 

It remains to to observe that $M'$ has at most 40 elements. Hence, \escir{} can be solved for $(M',\mathcal{X},P,\mathcal{Z})$ in time $\Oh(1)$ by brute force.
\end{proof}

Next, we consider graphic matroids. Recall that Bj{\"{o}}rklund, Husfeldt and Taslaman~\cite{BjorklundHT12}  proved that a shortest cycle that goes through a given set of $k$ vertices or edges in a graph can be found  in time 
$2^k\cdot n^{\Oh(1)}$. The currently best deterministic algorithm that finds a cycle that goes through a given set of $k$ vertices or edges was given by Kawarabayashi in~\cite{Kawarabayashi08}. 
We show that these results can be applied to solve \escir{}. 

\begin{lemma}\label{lem:graphic-scir}
\escir{} is \classFPT{} on graphic matroids when parameterized by $|\mathcal{X}|$.
\end{lemma}

\begin{proof}
Let $(M,\mathcal{X},P,\mathcal{Z})$ be an instance of \escir{}, where $M$ is a graphic matroid.  We find $G$ such that $M$ is isomorphic to $M(G)$ using the results of Seymour~\cite{Seymour81} and assume that $M=M(G)$. 

First, we show how to solve the problem for the case $\mathcal{Z}=\emptyset$ and then explain how to modify the algorithm if $\mathcal{Z}\neq \emptyset$.
Because the sets of $\mathcal{X}$ have sizes 2 or 3, $|P(X)|\leq 6$ for $X\in \mathcal{X}$ and
there is at most $6^{|\mathcal{X}|}$ possibilities to guess sets $S(X)\in P(X)$ of representatives of the elements $X\in\mathcal{X}$ in $C$. 
For each guess, let $T=\cup_{X\in \mathcal{X}}S(X)$.  Consider the graph $G'$  obtained from $G$ by the deletion of the elements of 
$(\cup_{X\in\mathcal{X}}X)\setminus T$. 
Clearly, $(M,\mathcal{X},P,\mathcal{Z})$ has a solution corresponding to the considered guess of sets $S(X)$ if and only if $G'$ has a cycle that goes through all the edges of $T$. To find such a cycle, we can apply the results of~\cite{BjorklundHT12} or~\cite{Kawarabayashi08}. 
If $\mathcal{Z}=(Z,t)$, we use Lemma~\ref{lem:gr-cogr} i). We additionally find a vertex $v$ of the cycle of $G$ induced by $Z$ that is not incident to the specified element $t$.
By Lemma~\ref{lem:gr-cogr} i),  $(M,\mathcal{X},P,\mathcal{Z})$ has a solution corresponding to the considered guess of sets $S(X)$ if and only if $G'$ has a cycle that goes through all the edges of $T\cup\{t\}$ and avoids $v$. To find such a cycle, we again can apply the results of~\cite{BjorklundHT12} or~\cite{Kawarabayashi08}. 

Since we consider at most  $6^{|\mathcal{X}|}$ guesses of  sets $S(X)\in P(X)$ and, for each guess, $|T|\leq 2|\mathcal{X}|$, we conclude that the algorithm runs in \classFPT{} time.
\end{proof}

For cographic matroids, we obtain the following lemma using the results of Robertson and Seymour~\cite{RobertsonS-GMXIII}.

\begin{lemma}\label{lem:cographic-scir}
\escir{} is \classFPT{} on cographic matroids when parameterized by $|\mathcal{X}|$.
\end{lemma}

\begin{proof}
Let $(M,\mathcal{X},P,\mathcal{Z})$ be an instance of \escir{}, where $M$ is a cographic matroid.  Using the results of Seymour~\cite{Seymour81}, we can in polynomial time find a graph $G$ such that $M$ is isomorphic to the bond matroid $M^*(G)$. We assume that $M=M^*(G)$. We can assume without loss of generality that $G$ is connected. 
Recall that to solve \escir{}, we have to check whether there is a cut $(A,B)$ of $G$ such that $G[A]$ and $G[B]$ are connected and $C=E(A,B)$ satisfies the requirements of \escir. 

Because the sets of $\mathcal{X}$ have sizes 2 or 3, $|P(X)|\leq 6$ for $X\in \mathcal{X}$ and
there is at most $6^{|\mathcal{X}|}$ possibilities to guess sets $S(X)\in P(X)$ of representatives of the elements $X\in\mathcal{X}$ in $C$. 
For each guess, let $T=\cup_{X\in \mathcal{X}}S(X)$.  If $\mathcal{Z}=(Z,t)$, then we additionally include $t$ in $T$.
Consider the graph $G'$  obtained from $G$ by the contraction of the elements of $(\cup_{X\in\mathcal{X}}X)\setminus T$. 

If there is $e\in T$ that is a loop of $G'$,  then $(M,\mathcal{X},Z,\mathcal{P})$ is a no-instance for the guess, since there is no minimal cut containing $e$. Assume that the edges of $T$ are not loops. 
We guess the placement of the end-vertices of the edges of $T$ in $A$ and $B$ considering at most $2^{|T|}$ possibilities. Let $T_A$ be the set of end-vertices guessed to be in $A$,  and let $T_B$ be the set of end-vertices in $B$. If $\mathcal{Z}=(Z,t)$, then we additionally put the end-vertices of the edges of $Z\setminus\{t\}$ in $T_B$ using Lemma~\ref{lem:gr-cogr} ii).
Now we have to check whether there is a partition $(A,B)$ of $V(G)$ such that $T_A\subseteq A$, $T_B\subseteq B$, and 
$G[A]$ and $G[B]$ are connected. By the celebrated results of Robertson and Seymour about disjoint paths, one can find in \classFPT{}-time with the parameter $|T_A|+|T_B|$ disjoint sets of vertices $A'$ and $B'$ containing $T_A$ and $T_B$ respectively such that $G[A']$ and $G[B']$ are connected  if such sets exist. If there are no such sets $A'$ and $B'$, we conclude that there is no partition $(A,B)$ with the required properties for the considered guess of $T_A$ and $T_B$. Otherwise, we extend $A'$ and $B'$ to the partition of $V(G)$ by the exhaustive applying the following rule: if there is $v\in V(G)\setminus(A'\cup B')$ that is adjacent to a vertex of $A'$ or $B'$, then put $v$ in $A'$ or $B'$ respectively. Clearly, we always obtain a partition of $V(G)$, because $G$ is connected. 

Since we consider at most  $6^{|\mathcal{X}|}$ guesses of  sets $S(X)\in P(X)$ and, for each guess, $|T|\leq 2|\mathcal{X}|$ and $|T_A|+|T_B|\leq 4|\mathcal{T}|+4$, we conclude that the algorithm runs in \classFPT{} time.
\end{proof}

\subsection{Proof of Theorem~\ref{thm:scir}}\label{sec:scir-reg} 
Now we are ready to give an algorithm for \scir{} on regular matroids. Let $(M,T)$ be an instance of \scir{}, where $M$ is regular. We consider it to be an instance  $(M,\mathcal{X},P,\mathcal{Z})$ of \escir, where $\mathcal{X}=\{\{t\}\mid t\in T\}$, 
$P(X)=X$ for $X\in \mathcal{X}$, and $\mathcal{Z}=\emptyset$.
If $M$ can be obtained from $R_{10}$ by the addition of parallel elements or 
is graphic or cographic, we solve the problem directly using Lemmas~\ref{lem:r10-scir}--\ref{lem:cographic-scir}. Assume that it is not the case. 
Using Theorem~\ref{thm:decomp-good},  we find a conflict tree $\mathcal{T}$. 
Recall that the set of nodes of $\mathcal{T}$ is the collection of basic matroids $\mathcal{M}$ and the edges correspond to extended $1$-, $2-$ and 3-sums.  The key observation is that $M$ can be constructed from $\mathcal{M}$ by performing the sums corresponding to the edges of $\mathcal{T}$ in an arbitrary order.  We select an arbitrarily node $r$ of $\mathcal{T}$ containing an element of $T$ as a root.
Our algorithm is based on performing bottom-up traversal of the tree $\mathcal{T}$. We exhaustively apply reduction rules that remove leaves of $\mathcal{T}$ until we obtain a basic case for which we can apply Lemmas~\ref{lem:r10-scir}--\ref{lem:cographic-scir}.

We say that an instance $(M,\mathcal{X},P,\mathcal{Z})$ of \escir{} is \emph{consistent (with respect to $\mathcal{T}$)} if $\mathcal{Z}=\emptyset$ and for any $X\in\mathcal{X}$, $X\in E(M)$ for some $M\in\mathcal{M}$. Clearly, the instance obtained from the original input instance $(M,T)$ of \scir{} is consistent. Our reduction rules keep this property.

Let $M_\ell\in\mathcal{M}$ be a matroid that is a leaf of $\mathcal{T}$. Denote by $M_s$ its adjacent sub-leaf.  
We construct reduction rules depending on whether $M_\ell$ is 1, 3 or 3 leaf.

Throughout this section, we say that a reduction rule is \emph{safe} if it either correctly solves the problem or returns an equivalent instance of \escir{} together with corresponding conflict tree of the obtained matroid that is consistent and the value of the parameter does not increase.  

\begin{reduction}[{\bf $1$-Leaf reduction rule}]\label{rule:1-leaf-scir}
If $M_\ell$ is a 1-leaf, then do the following.
\begin{itemize}
\item[i)] If there is $X\in\mathcal{X}$ such that $X\in E(M_\ell)$, then stop and return a no-answer,
\item[ii)] Otherwise, delete $M_\ell$ from $\mathcal{T}$ and denote by $T'$ the obtained conflict tree. Return the instance  $(M',\mathcal{X},P,\emptyset)$ and solve it using the conflict tree $\mathcal{T}'$, where $M'$ is the matroid defined by $\mathcal{T}'$.
\end{itemize}
\end{reduction}

Since the root matroid contains at least one set of $\mathcal{X}$, Lemma~\ref{lem:circ} i) immediately implies the following lemma.

\begin{lemma}\label{lem:1-leaf-scir}
Reduction Rule~\ref{rule:1-leaf-scir}  is safe and can be implemented to run in time polynomial in $|E(M)|$. 
\end{lemma}

\begin{reduction}[{\bf $2$-Leaf reduction rule}]\label{rule:2-leaf-scir}
If $M_\ell$ is a 2-leaf, then let $\{e\}=E(M_\ell)\cap E(M_s)$ and do the following. 
\begin{itemize}
\item[i)] If there is no $X\in\mathcal{X}$ such that $X\in E(M_\ell)$, then check whether there is a circuit of $M_\ell$ containing $e$. If there is no such a circuit, then delete $e$ from $M_s$. Delete $M_\ell$ from $\mathcal{T}$ and denote by $T'$ the obtained conflict tree. Return the instance  $(M',\mathcal{X},P,\emptyset)$ and solve it using the conflict tree $\mathcal{T}'$, where $M'$ is the matroid defined by $\mathcal{T}'$. 
\item[ii)] Otherwise, if there is  $X\in\mathcal{X}$ such that $X\in E(M_\ell)$, consider $\mathcal{X}_\ell=\{X\in\mathcal{X}\mid X\subseteq E(M_\ell)\}\cup\{\{e\}\}$. Set $P_\ell(X)=P(X)$ for $X\in\mathcal{X}_\ell$ such that $X\neq\{e\}$, and set $P_\ell(\{e\})=\{e\}$.  Solve \escir{} for $(M_\ell,\mathcal{X}_\ell,P_\ell,\emptyset)$. 
If $(M_\ell,\mathcal{X}_\ell,P_\ell,\emptyset)$ is a no-instance, then stop and return a no-answer. Otherwise, do the following. 
Set $\mathcal{X}'=\{X\in\mathcal{X}\mid X\not\subseteq E(M_\ell)\}\cup\{\{e\}\}$. Set $P'(X)=P(X)$ for $X\in\mathcal{X}'$ such that $X\neq\{e\}$, and set $P'(\{e\})=\{e\}$.  Delete $M_\ell$ from $\mathcal{T}$ and denote the obtained conflict tree by $\mathcal{T}'$. Let $M'$ be the matroid defined by $\mathcal{T}'$.
Return the instance  $(M',\mathcal{X}',P',\emptyset)$ and solve it using the conflict tree $\mathcal{T}'$.
\end{itemize}
\end{reduction}

\begin{lemma}\label{lem:2-leaf-scir}
Reduction Rule~\ref{rule:2-leaf-scir}  is safe and can be implemented to run in time $f(\mathcal{X})\cdot  n^{\cO(1)}$ for some function $f$ of $\mathcal{X}$ only. 
\end{lemma}

\begin{proof}
Clearly, if the rule returns a new instance, then it is consistent 
with respect to $\mathcal{T}'$ and the parameter does not increase.

We show that the rule either correctly solves the problem or returns an equivalent instance. Denote by $\hat{M}$ the matroid defined by the conflict tree obtained from $\mathcal{T}$ by the deletion of the node $M_\ell$. Clearly, $M=\hat{M}\oplus_2 M_\ell$.

Suppose that $(M,\mathcal{X},P,\mathcal{Z})$ is a consistent yes-instance. We prove that the rule returns a yes-instance. Denote by $C$ a circuit of $M$ that is a solution for $(M,\mathcal{X},P,\mathcal{Z})$.
We consider two cases corresponding to the cases i) and ii) of the rule.

\medskip
\noindent
{\bf Case 1.}  There is no $X\in\mathcal{X}$ such that $X\in E(M_\ell)$. If $C\subseteq E(\hat{M})$, then by Lemma~\ref{lem:circ} ii), $C$ is a circuit of $M'$ constructed by the rule that is either $\hat{M}$ or the matroid obtained by from $\hat{M}$ by the deletion of $e$, because $e\notin C$. Suppose that $C\cap E(M_\ell)\neq\emptyset$. Then $C=C_1\bigtriangleup C_2$, where $C_1\in \mathcal{C}(\hat{M})$, $C_2\in\mathcal{C}(M_2)$ and $e\in C_1\cap C_2$ by Lemma~\ref{lem:circ} ii). Because $C_2$ is a circuit of $M_2$ containing $e$, we do not delete $e$ from $M_s$ and, therefore, $C_1$ is a circuit of $M'=\hat{M}$ constructed by the rule in this case. It remains to observe that $C_1$ is a solution for $(M',\mathcal{X},P,\emptyset)$. Hence, $(M',\mathcal{X},P,\emptyset)$ is a yes-instance.

\medskip
\noindent
{\bf Case 2.}  There is  $X\in\mathcal{X}$ such that $X\in E(M_\ell)$. Then by Lemma~\ref{lem:circ} ii), $C=C_1\bigtriangleup C_2$, where $C_1\in \mathcal{C}(\hat{M})$, $C_2\in\mathcal{C}(M_2)$ and $e\in C_1\cap C_2$. We have that $C_2$ is a solution for $(M_\ell,\mathcal{X}_\ell,P_\ell,\emptyset)$ and the algorithm does not stop. Also we have that $C_1$ is a solution for $(M',\mathcal{X}',P',\emptyset)$, i.e., the rule returns a yes-instance.

\medskip
Suppose now that the instance constructed by the rule is a yes-instance with a solution $C'$. We show that the original instance $(M,\mathcal{X},P,\mathcal{Z})$ is a  yes-instance. We again consider two cases.

\medskip
\noindent
{\bf Case 1.} The new instance is constructed by Rule~\ref{rule:2-leaf-scir} i). If $e\notin C'$, then $C'$ is a circuit of $M$ by Lemma~\ref{lem:circ} ii) and, therefore, $C'$ is a solution for  $(M,\mathcal{X},P,\mathcal{Z})$, that is,  $(M,\mathcal{X},P,\mathcal{Z})$ is a yes-instance. Assume that $e\in C'$. In this case, $e$ was not deleted by the rule from $M_s$. Hence, there is a circuit $C''$ of $M_\ell$ containing $e$. By Lemma~\ref{lem:circ} ii), $C=C'\bigtriangleup C''$ is a circuit of $M$. We have that $C$ is a solution for $(M,\mathcal{X},P,\mathcal{Z})$ and it is a yes-instance.

\medskip
\noindent
{\bf Case 2.} The new instance is constructed by Rule~\ref{rule:2-leaf-scir} ii). In this case, $(M_\ell,\mathcal{X}_\ell,P_\ell,\emptyset)$ is a yes-instance and there is a solution $C''$ for it. Notice that $e\in C'\cap C''$. We have that $C=C'\bigtriangleup C''$ is a circuit of $M$ by Lemma~\ref{lem:circ} ii). We have that $C$ is a solution for $(M,\mathcal{X},P,\mathcal{Z})$ and, therefore, 
 $(M,\mathcal{X},P,\mathcal{Z})$ is a yes-instance.

\medskip
We proved that the rule is safe. To evaluate the running time, notice first that we can check existence of a circuit of $M_\ell$ containing $e$ in Rule~\ref{rule:2-leaf-scir} i) in polynomial time either directly or using the straightforward observation that we have an instance of \scir{} with $T=\{e\}$ and can apply Lemmas~\ref{lem:r10-scir}--\ref{lem:cographic-scir} depending on the type of $M_\ell$. The problem for $(M_\ell,\mathcal{X}_\ell,P_\ell,\emptyset)$ in Rule~\ref{rule:2-leaf-scir} ii) can be solved in \classFPT{} time by Lemmas~\ref{lem:r10-scir}--\ref{lem:cographic-scir} depending on the type of $M_\ell$, because $|\mathcal{X}_\ell|\leq|\mathcal{X}|$. 
\end{proof}

\begin{reduction}[{\bf $3$-Leaf reduction rule}]\label{rule:3-leaf-scir}
If $M_\ell$ is a 3-leaf, then let $Z=E(M_\ell)\cap E(M_s)=\{e_1,e_2,e_3\}$ and do the following. 
\begin{itemize}
\item[i)] If there is no $X\in\mathcal{X}$ such that $X\in E(M_\ell)$, then for each $i\in\{1,2,3\}$, solve \escir{} for the instance $(M_\ell,\emptyset,\emptyset,(Z,e_i))$, and if $(M_\ell,\emptyset,\emptyset,(Z,e_i))$ is a no-instance, then delete $e_i$ from $M_s$. Delete $M_\ell$ from $\mathcal{T}$ and denote by $T'$ the obtained conflict tree. Return the instance  $(M',\mathcal{X},P,\emptyset)$ and solve it using the conflict tree $\mathcal{T}'$, where $M'$ is the matroid defined by $\mathcal{T}'$. 
\item[ii)] Otherwise, if there is  $X\in\mathcal{X}$ such that $X\in E(M_\ell)$, set $\mathcal{X}_\ell=\{X\in\mathcal{X}\mid X\subseteq E(M_\ell)\}$ and  $P_\ell(X)=P(X)$ for $X\in\mathcal{X}_\ell$.
We construct the set $R$ of subsets of $Z$ as follows. Initially, $R=\emptyset$.
\begin{itemize}
\item For $i\in \{1,2,3\}$, solve \escir{} for the instance $(M_\ell,\mathcal{X}_\ell,P_\ell,(Z,e_i))$, and if we get a yes-instance, then add $\{e_i\}$ in $R$.
\item For $i\in \{1,2,3\}$, solve \escir{} for the instance $(M_\ell,\mathcal{X}_\ell',P_\ell^{(i)},\emptyset)$, where $\mathcal{X}_\ell'=\mathcal{X}_\ell\cup\{Z\}$ and $P_\ell^{(i)}(X)=P_\ell(X)$ for $X\in\mathcal{X}_\ell$ and $L_\ell^{(i)}(Z)=\{e_i\}$. If we  get a yes instance, then add $Z\setminus\{e_i\}$ in $R$.
\end{itemize}
If $R=\emptyset$, then  stop and return a no-answer. Otherwise, do the following. 
Set $\mathcal{X}'=\{X\in\mathcal{X}\mid X\not\subseteq E(M_\ell)\}\cup\{Z\}$. Set $P'(X)=P(X)$ for $X\in\mathcal{X}'$ such that $X\neq Z$, and set $P'(Z)=R$.  Delete $M_\ell$ from $\mathcal{T}$ and denote the obtained conflict tree by $\mathcal{T}'$. Let $M'$ be the matroid defined by $\mathcal{T}'$.
Return the instance  $(M',\mathcal{X}',P',\emptyset)$ and solve it using the conflict tree $\mathcal{T}'$.
\end{itemize}
\end{reduction}

\begin{lemma}\label{lem:3-leaf-scir}
Reduction Rule~\ref{rule:3-leaf-scir}  is safe and
and can be implemented to run in time $f(\mathcal{X})\cdot  n^{\cO(1)}$ for some function $f$ of $\mathcal{X}$ only. 
\end{lemma}

\begin{proof}
Clearly, if the rule returns a new instance, then it is consistent with respect to $\mathcal{T}'$ and the parameter does not increase.

We show that the rule either correctly solves the problem or returns an equivalent instance. Denote by $\hat{M}$ the matroid defined by the conflict tree obtained from $\mathcal{T}$ by the deletion of the node $M_\ell$. Clearly, $M=\hat{M}\oplus_2 M_\ell$.

Suppose that $(M,\mathcal{X},P,\mathcal{Z})$ is a consistent yes-instance. We prove that the rule returns a yes-instance. Denote by $C$ a circuit of $M$ that is a solution for $(M,\mathcal{X},P,\mathcal{Z})$.
We consider two cases corresponding to the cases i) and ii) of the rule.

\medskip
\noindent
{\bf Case 1.}  There is no $X\in\mathcal{X}$ such that $X\in E(M_\ell)$. If $C\subseteq E(\hat{M})$, then by Lemma~\ref{lem:circ} iii), $C$ is a circuit of $M'$ constructed by the rule that is obtained by from $\hat{M}$ by the deletion of some elements of $Z$, because $Z\cap C=\emptyset$. 
Suppose that $C\cap E(M_\ell)\neq\emptyset$. Then, by Lemma~\ref{lem:circ} iii),
$C=C_1\bigtriangleup C_2$, where $C_1\in \mathcal{C}(\hat{M})$, $C_2\in\mathcal{C}(M_\ell)$, $C_1\cap Z=C_2\cap Z=\{e_i\}$ for some $i\in\{1,2,3\}$, and 
$C_1\bigtriangleup Z$ is a circuit of $\hat{M}$ or $C_2\bigtriangleup Z$ is a circuit of $M_\ell$.

Suppose that $C_2\bigtriangleup Z$ is a circuit of $M_\ell$. Then $(M_\ell,\emptyset,\emptyset,(Z,e_i))$ is a yes-instance and, therefore, $e_i\in E(M')$. Hence,  $C_1$ is a circuit of $M'$ constructed by the rule. We have that $C_1$ is a solution for $(M',\mathcal{X},P,\emptyset)$. Hence, $(M',\mathcal{X},P,\emptyset)$ is a yes-instance.

Assume now that $C_2\bigtriangleup Z$ is a not circuit of $M_\ell$. By Lemma~\ref{lem:circ-triangle}, $C_2$ is a disjoint union of two circuits $C_2^{(1)}$ and $C_2^{(2)}$ of $M_2$ containing $e_h,e_j\in Z\setminus\{e_i\}$, and  $C_2^{(1)}\bigtriangleup Z$ and $C_2^{(2)}\bigtriangleup Z$ are circuits of $M_\ell$. 
Then $(M_\ell,\emptyset,\emptyset,(Z,e_h))$ and $(M_\ell,\emptyset,\emptyset,(Z,e_h))$
are yes-instances and, therefore, $e_h,e_j\in E(M')$. Consider $C_1'=C_1\bigtriangleup Z$. Because  $C_2\bigtriangleup Z$ is a not circuit of $M_\ell$, $C_1'$ is a circuit of $\hat{M}$. Since $e_h,e_j\in E(M')$, we have that $C_1'$ is a solution for $(M',\mathcal{X},P,\emptyset)$. Hence, $(M',\mathcal{X},P,\emptyset)$ is a yes-instance.

\medskip
\noindent
{\bf Case 2.}  There is  $X\in\mathcal{X}$ such that $X\in E(M_\ell)$. We have that $C=C_1\bigtriangleup C_2$, where $C_1\in \mathcal{C}(\hat{M})$, $C_2\in\mathcal{C}(M_\ell)$, $C_1\cap Z=C_2\cap Z=\{e_i\}$ for some $i\in\{1,2,3\}$, and $C_1\bigtriangleup Z$ is a circuit of $\hat{M}$ or $C_2\bigtriangleup Z$ is a circuit of $M_\ell$. 

Suppose that $C_2\bigtriangleup Z$ is a circuit of $M_\ell$. Then $(M_\ell,\mathcal{X}_\ell,P_\ell,(Z,e_i))$ is a yes-instance and, therefore, $\{e_i\}\in R$. Since $R\neq\emptyset$, the algorithm does not stop.  Also we have that $C_1$ is a solution for $(M',\mathcal{X}',P',\emptyset)$, i.e., the rule returns a yes-instance.

Assume now that $C_2\bigtriangleup Z$ is not a circuit of $M_\ell$. Then $(M_\ell,\mathcal{X}_\ell',P_\ell^{(i)},\emptyset)$ is a yes-instance and, therefore, $Z\setminus \{e_i\}\in R$. Since $R\neq\emptyset$, the algorithm does not stop.  Consider $C_1'=C_1\bigtriangleup Z$. Notice that $C_1'$ is a circuit of $\hat{M}$. 
We obtain that $C_1'$ is a solution for $(M',\mathcal{X}',P',\emptyset)$, i.e., the rule returns a yes-instance.

\medskip
Suppose now that the instance constructed by the rule is a yes-instance with a solution $C'$. We show that the original instance $(M,\mathcal{X},P,\mathcal{Z})$ is a  yes-instance. We again consider two cases.

\medskip
\noindent
{\bf Case 1.} The new instance is constructed by Rule~\ref{rule:3-leaf-scir} i). 

If $C'\cap Z=\emptyset$, then $C'$ is a circuit of $M$ by Lemma~\ref{lem:circ} iii) and, therefore, $C'$ is a solution for  $(M,\mathcal{X},P,\mathcal{Z})$, that is,  $(M,\mathcal{X},P,\mathcal{Z})$ is a yes-instance. 

Suppose that $C'\cap Z=\{e_i\}$ for some $i\in\{1,2,3\}$. Then, by the construction of the rule, there is a circuit $C''$ of $M_\ell$ such that $C''\cap Z=\{e_i\}$ and 
$C''\bigtriangleup Z$ is a circuit.  By Lemma~\ref{lem:circ} iii), $C=C'\bigtriangleup C''$ is a circuit of $M$. We have that $C$ is a solution for $(M,\mathcal{X},P,\mathcal{Z})$ and it is a yes-instance.

Assume that $C'\cap Z=\{e_h,e_j\}$ for some distinct $h,j\in\{1,2,3\}$. Let $e_i$ be the element of $Z$ distinct from $e_h$ and $e_j$. We have that $M_\ell$ has two circuits $C_h$ and $C_j$ such that $C_h\cap Z=\{e_h\}$, $C_j\cap Z=\{e_j\}$. Then $C_h\bigtriangleup C_j\bigtriangleup Z$ is a cycle of $M_\ell$ by Observation~\ref{obs:symm}, and this cycle contains a circuit $C_i$ such that $C_i\cap Z=\{e_i\}$. Consider $C''=C'\bigtriangleup Z$. By Lemma~\ref{lem:circ-triangle-a}, $C''$ is a circuit of $\hat{M}$ and $C''\bigtriangleup Z$ is a circuit.
By Lemma~\ref{lem:circ} iii), we conclude that $C=C''\bigtriangleup C_i$ is a solution for $(M,\mathcal{X},P,\mathcal{Z})$ and, therefore, $(M,\mathcal{X},P,\mathcal{Z})$ is a yes-instance.

\medskip
\noindent
{\bf Case 2.} The new instance is constructed by Rule~\ref{rule:2-leaf-scir} ii). In this case, $C''\cap Z\in P'(Z)=R$. Recall that $R$ contains sets of size 1 or 2.

Suppose that  $C'\cap Z=\{e_i\}$ for some $i\in\{1,2,3\}$. Then, by the construction of the rule, there is a solution $C''$ for the instance $(M_\ell,\mathcal{X}_\ell,P_\ell,(Z,e_i))$. Notice that 
$C''\cap Z=\{e_i\}$ and $C''\bigtriangleup Z$ is a circuit of $M_\ell$.  By Lemma~\ref{lem:circ} iii), $C=C'\bigtriangleup C''$ is a circuit of $M$. We have that $C$ is a solution for $(M,\mathcal{X},P,\mathcal{Z})$ and it is a yes-instance.

Assume that $C'\cap Z=\{e_h,e_j\}$ for some distinct $h,j\in\{1,2,3\}$. Let $e_i$ be the element of $Z$ distinct from $e_h$ and $e_j$. There is a solution $C''$ for 
$(M_\ell,\mathcal{X}_\ell',P_\ell^{(i)},\emptyset)$. Recall that $C''\cap Z=\{e_i\}$. Consider $C'''=C'\bigtriangleup Z$. By Lemma~\ref{lem:circ-triangle-a}, $C'''$ is a circuit of $\hat{M}$ and $C'''\bigtriangleup Z$ is a circuit. Since $C'''\cap Z=\{e_i\}$, we obtain that $C=C'''\bigtriangleup C''$ is a circuit of $M$. It remains to observe that $C$ is a solution for $(M,\mathcal{X},P,\mathcal{Z})$ and it is a yes-instance.

\medskip
We proved that the rule is safe. To evaluate the running time, notice first that we can check existence of a circuit of $M_\ell$ containing each $e_i$ in Rule~\ref{rule:3-leaf-scir} ii) in polynomial time using Lemmas~\ref{lem:r10-scir}--\ref{lem:cographic-scir} depending on the type of $M_\ell$. The problems for $(M_\ell,\mathcal{X}_\ell,P_\ell,(Z,e_i))$ and  $(M_\ell,\mathcal{X}_\ell',P_\ell^{(i)},\emptyset)$
in Rule~\ref{rule:3-leaf-scir} i) can be solved in \classFPT{} time by Lemmas~\ref{lem:r10-scir}--\ref{lem:cographic-scir} depending on the type of $M_\ell$, because $|\mathcal{X}_\ell|<|\mathcal{X}_\ell'|\leq|\mathcal{X}|$. 
\end{proof}

To complete the proof of Theorem~\ref{thm:scir}, it remains to observe that $\mathcal{M}$ and the corresponding conflict tree $\mathcal{T}$ can be constructed in polynomial time by Theorem~\ref{thm:decomp-good}, and then we apply the reduction rules at most $|V(\mathcal{T})|-1$ times until we obtain an instance of \escir{} for a matroid of one of basic types and solve the problem using Lemmas~\ref{lem:r10-scir}--\ref{lem:cographic-scir}.

\section{Lower bounds and open questions}\label{sec:conclusion}
In this paper we gave \classFPT algorithms for  \probWMSC and
 \probSCIR for regular matroids.  We conclude with a number of open algorithmic questions about circuits in matroids. We also discuss here certain algorithmic limitations for extending our results. 
 
 \medskip\noindent\textbf{Larger matroid classes.}
 The first natural question is  {whether our results can be extended to other classes of matroids}? There is no hope (of course up to certain complexity assumptions) that our  results can be extended to binary matroids. 
 Downey et al. proved in~\cite{DowneyFVW99} that the following  problem is \classW{1}-hard being parameterized by $k$. (We refer to the book of Downey and Fellows \cite{DowneyFbook13} for the definition of W-hierarchy.) 
In the 
 {\textsc{Maximum-Likelihood Decoding}}  problem we are given a binary $n\times m$ matrix $A$, a target binary $n$-element vector $\vec{s}$, and a positive integer $k$.
 The question is whether  
 there is a set of at most $k$ columns of $A$ that sum to $\vec{s}$?  
 As it was observed by Gavenciak et al. \cite{GavenciakKO12}, 
 the result of Downey et al. immediately implies the following proposition.
 
\begin{proposition}[\cite{GavenciakKO12}]\label{prop:W-h-bin}
\wmsc{} is \classW{1}-hard on binary matroids with unit-weights elements when parameterized by $\ell$ even when  $|T|=1$.
\end{proposition}
Let us note that \wmsc with  $|T|=0$  on binary matroids  is equivalent to \textsc{Even Set}, which  parameterized complexity is a long standing open question, see e.g. 
\cite{DowneyFbook13}.

However Proposition~\ref{prop:W-h-bin} does not rule out a possibility that our results can be extended from the class of regular matroids to  any proper minor-closed class of binary, and even more generally, representable over some finite field,  matroids. It is very likely that the powerful structural theorems obtained by Geelen et al.  in order to settle Rota's conjecture,  see \cite{GeelenB14} for further discussions, can shed some light on this question.  

Solving both problems on transversal matroids is another interesting problem. 

 \medskip\noindent\textbf{Stronger parameterization.}
 Bj{\"{o}}rklund et al.   in~\cite{BjorklundHT12} gave a randomized algorithm that finds a shortest cycle through a given set $T$ of   vertices or edges in a graph in time 
$2^{|T|}\cdot n^{\Oh{1}}$. Hence  \wmsc{} parameterized by $w(T)$ is (randomized) \classFPT{} on graphic matroids if the weights are encoded in unary. Unfortunately, it is possible to show  that \wmsc is \classW{1}-hard already on cographic matroids for this parameterization.

\begin{theorem}\label{prop:W-h-cogr}
\wmsc{} is \classW{1}-hard on cographic matroids with unit-weights elements when parameterized by $|T|$.
\end{theorem}

\begin{proof}
We reduce the following  variant of the \textsc{Multicolored Clique} problem.
In the {\textsc{Regular Multicolored Clique}}  we are given a 
  regular graph $G$, a positive integer parameter $k$, and a partition $V_1,\ldots,V_k$ of $V(G)$. 
The task is to decide whether 
  $G$ have a clique $K$ such that $|V_i\cap K|=1$ for $i\in\{1,\ldots,k\}$. 
  {\textsc{Regular Multicolored Clique}}  parameterized by $k$ was shown  to be \classW{1}-hard by Cai in~\cite{Cai08}.

Let $(G,k,V_1,\ldots,V_k)$ be an instance of \textsc{Regular Multicolored Clique}, and assume that $G$ is a $d$-regular $n$-vertex graph. Assume without loss of generality that $k<d<n-1$.
We construct the graph $H$ as follows.
\begin{itemize}
\item Construct a copy of $G$.
\item For each $i\in\{1,\ldots,k\}$, construct a vertex $v_i$ and edges $v_iu$ for $u\in V_i$.
\item Construct $n$ pairwise adjacent vertices $x_1,\ldots,x_n$ and make them adjacent to the vertices of $G$.
\item Construct $p=2n^2$ pairwise adjacent vertices $y_1,\ldots,y_p$ and make each of them adjacent to $x_1,\ldots,x_n$.
\item Construct edges $y_1v_1,\ldots,y_1v_k$ and set $T=\{y_1v_1,\ldots,y_1v_k\}$.
\end{itemize}
We put $\ell=n+(n+d-k+1)k$.

We claim that  $(G,k,V_1,\ldots,V_k)$ is a yes-instance of \textsc{Regular Multicolored Clique} if and only if $H$ has a minimal cut-set $C$ of size at most $\ell$ such that $T\subseteq C$. 

Suppose that $K$ is a clique in $G$ with $|V_i\cap K|=1$ for $i\in\{1,\ldots,k\}$. Consider the partition $(A,\overline{A})$ of $V(G)$ with $A=\{v_1,\ldots,v_k\}\cup K$. It is straightforward to verify that $H[A]$ and $H[\overline{A}]$ connected. Therefore $C=E(A,\overline{A})$ is a minimal cut-set. The vertices $v_1,\ldots,v_k$ have $n-k$ neighbors in $V(G)\cap\overline{A}$ in total and all their neighbors are distinct. Also each $v_i$ is adjacent to $y_1\in \overline{A}$.
Since $G$ is $d$-regular, each vertex $u\in K$ has $d-k+1$ neighbors in $V(G)\cap \overline{A}$ and $n$ neighbors $x_1,\ldots, x_n$ among the remaining vertices of $\overline{A}$. Hence, $|C|=(n-k)+k+(n+d-k+1)k  =\ell$.

Assume now that $H$ has a minimal cut-set $C$ of size at most $\ell$ such that $T\subseteq C$. Let $(A,\overline{A})$ be the partition of $V(H)$ with $E(A,\overline{A})=C$. We also assume that $y_1\in \overline{A}$. Then $v_1,\ldots,v_k\in A$. 

First, we show that $x_i,y_j\in \overline{A}$ for $i\in\{1,\ldots,n\}$ and $j\in\{1,\ldots,p\}$. To obtain a contradiction, assume that at least one of these vertices is in $A$. Because
$\{x_1,\ldots,x_n\}\cup \{y_1,\ldots,y_n\}$ is a clique of size $2n^2+n$ and $T\subseteq E(A,\overline{A})$, we have that 
$|E(A,\overline{A})|\geq 2n^2+n-1+k>n+(n+d-k+1)k=\ell$ contradicting $|E(A,\overline{A})|\leq \ell$.

Because $H[A]$ is connected and $v_1,\ldots,v_k\in A$, there is $u_i\in V_i$ such that $u_i\in A$ for each $i\in\{1,\ldots,k\}$. Let $A'=\{v_1,\ldots,v_k\}\cup\{u_1,\ldots,u_k\}$. 
The vertices $v_1,\ldots,v_k$ have $n-k$ neighbors in total in $V(G)\cap\overline{A'}$ and all their neighbors are distinct. Also each $v_i$ is adjacent to $y_1\in \overline{A'}$.
Since $G$ is $d$-regular, each vertex $u_i$ has at least $d-k+1$ neighbors in $V(G)\cap \overline{A'}$, and all the vertices $u_1,\ldots,u_k$ are incident to $(d-k+1)d$ edges of $G$ with exactly one end-vertex  in $A'$ if and only if $\{u_1,\ldots,u_k\}$ is a clique of $G$. Also each vertex $u_i$ is adjacent to $x_1,\ldots,x_k$. Therefore,  $|E(A',\overline{A'})|\geq (n-k)+k+(n+d-k+1)k  =\ell$, and $|E(A',\overline{A'})|=\ell$ if and only if $\{u_1,\ldots,u_k\}$ is a clique of $G$.

 Since $x_i,y_j\in \overline{A}$ for $i\in\{1,\ldots,n\}$ and $j\in\{1,\ldots,p\}$, $A'\setminus A\subseteq V(G)$. Because $G$ is $d$-regular and each vertex of $G$ is adjacent to exactly one vertex $v_i$ and the vertices $x_1,\ldots,x_n$, 
$\ell=|E(A,\overline{A})|\geq |E(A',\overline{A'})|+|A'\setminus A|(n-d-1)\geq |E(A',\overline{A'})|\geq \ell$. 
As $d<n-1$, we obtain that $A=A'$. Hence,  $\{u_1,\ldots,u_k\}$ is a clique of $G$.

To complete the proof, we observe that $H$ has a minimal cut-set $C$ of size at most $\ell$ such that $T\subseteq C$ if and only if $(M(H),w,T,\ell)$ is a yes-instance of \wmsc{} with the weight function $w(e)=1$ for $e\in E(H)$. 
\end{proof}

Interestingly, Theorem~\ref{prop:W-h-cogr} does not rule out a possibility that for a fixed numbers of terminals
\wmsc is still resolvable in polynomial time, or in other words that it is in \classXP parameterized by $|T|$. We conjecture that this is not the case. More precisely,  {is  \wmsc \classNP-complete on cographic matroids for a fixed number, say $|T|=3$, terminal elements?}

 \medskip\noindent\textbf{Other circuit problems.} We do not know if our technique could be adapted to solve the following variant of the spanning circuit problem. Given a regular matroid $M$ with a set of terminals,    decide whether  $M$ contains a circuit of size at least $\ell$ spanning  all terminals. We leave the complexity of this problem parameterized by $\ell$ open.

Another interesting variation of \probWMSC and
 \probSCIR is the problem where we seek for a  circuit of a given parity containing a given set of terminal elements $T$. For graphs (or graphic matroids), Kawarabayshi et al.  \cite{KawarabayashiLR10} proved that the problem is \classFPT parameterized by $|T|$. The complexity of this problem on cographic matroids is open.

\end{document}